\documentclass[a4paper,UKenglish,cleveref, autoref, thm-restate]{lipics-v2021}
\usepackage[T1]{fontenc}
\usepackage{lmodern}

\hideLIPIcs  

\title{Worbel: Aggregating Point Labels \texorpdfstring{\\ into Word Clouds}{}} 

\titlerunning{Worbel: Aggregating Point Labels into Word Clouds} 
\author{Sujoy Bhore}{Indian Institute of Science Education and Research, Bhopal, India}{sujoy.bhore@gmail.com}{https://orcid.org/0000-0003-0104-1659}{}

\author{Robert Ganian}{Algorithms and Complexity Group, TU Wien, Vienna, Austria}{rganian@gmail.com}{https://orcid.org/0000-0002-7762-8045}{}
  
\author{Guangping Li}{Algorithms and Complexity Group, TU Wien, Vienna, Austria}{guangping@ac.tuwien.ac.at}{https://orcid.org/0000-0002-7966-076X}{}

\author{Martin N\"ollenburg}{Algorithms and Complexity Group, TU Wien, Vienna, Austria}{noellenburg@ac.tuwien.ac.at}{https://orcid.org/0000-0003-0454-3937}{}

\author{Jules Wulms}{Algorithms and Complexity Group, TU Wien, Vienna, Austria}{jwulms@ac.tuwien.ac.at}{https://orcid.org/0000-0002-9314-8260}{}

\authorrunning{S. Bhore et al.} 

\Copyright{Sujoy Bhore, Robert Ganian, Guangping Li, Martin N\"ollenburg and Jules Wulms} 
\usepackage{mathtools}



\category{} 

\graphicspath{{./figs_arxiv/}}
\relatedversion{} 

\supplement{Source code, benchmark data generator and real-world data extractor at\\ \url{https://dyna-mis.github.io/geoWordle/}.}

\acknowledgements{We acknowledge support by the Austrian Science Fund (FWF) under grants P 31336 and Y1329 (Robert Ganian) and grant P 31119 (Guangping Li, Martin N\"{o}llenburg, and Jules Wulms)}

\nolinenumbers 

\newcommand{\NP}{\textsf{\textup{NP}}}
\newcommand{\RPFA}{\textsc{\textup{RPFA}}}  
\newcommand{\DBCR}{\textsc{DBCR}}
\newcommand{\III}{\mathcal{I}}

\newcommand{\mypar}[1]{\smallskip\noindent{\bfseries #1.}\quad}
\EventEditors{John Q. Open and Joan R. Access}
\EventNoEds{2}
\EventLongTitle{42nd Conference on Very Important Topics (CVIT 2016)}
\EventShortTitle{CVIT 2016}
\EventAcronym{CVIT}
\EventYear{2016}
\EventDate{December 24--27, 2016}
\EventLocation{Little Whinging, United Kingdom}
\EventLogo{}
\SeriesVolume{42}
\ArticleNo{23}
\begin{document}

\maketitle

\begin{abstract}
Point feature labeling is a classical problem in cartography and GIS that has been extensively studied for geospatial point data.
At the same time, word clouds are a popular visualization tool to show the most important words in text data which has also been extended to visualize geospatial data (Buchin et al. PacificVis 2016).

In this paper, we study a hybrid visualization, which combines aspects of word clouds and point labeling.
In the considered setting, the input data consists of a set of points grouped into categories and our aim is to place multiple disjoint and axis-aligned rectangles, each representing a category, such that they cover points of (mostly) the same category under some natural quality constraints.
In our visualization, we then place category names inside the computed rectangles to produce a labeling of the covered points which summarizes the predominant categories globally (in a word-cloud-like fashion) while locally avoiding excessive misrepresentation of points (i.e., retaining the precision of point labeling). 

We show that computing a minimum set of such rectangles is \NP-hard. 
Hence, we turn our attention to developing heuristics and exact \textsf{SAT} models to compute our visualizations.
We evaluate our algorithms quantitatively, measuring running time and quality of the produced solutions, on several artificial and real-world data sets. 
Our experiments show that the heuristics produce solutions of comparable quality to the \textsf{SAT} models while running much faster.
\end{abstract}
\section{Introduction}

Labeling graphical features in maps is a classical problem in cartography and geographic information systems (GIS), and has been extensively studied~\cite{kreveld-book,y-lal-72}. 
The map features to be labeled range from points and lines to areas.
Point labeling was introduced in computational cartography about 30 years ago~\cite{fw-ppwalm-91,ms-ccclp-91} and has since been studied in many different settings, e.g.,~\cite{AgarwalKS98,s-gaclp-01,hw-bmieipfpl-17,rr-cmmhcqplp-14}.
In many real-life applications, however, which deal with massive data, there can be many repetitions of labels. 
For instance, in twitter data, where hashtags indicate the topic of a tweet, or in categorical point-of-interest data, a few common topics appear many times.
Geographic visualizations of huge social network data, especially the ones used in anomaly detection and visual analytics~\cite{DBLP:conf/apvis/ThomBKWE12, DBLP:conf/ieeevast/MacEachrenJRPSMZB11} usually require information aggregation.

A \emph{tag map} is a visualization technique that achieves this by placing a categorical label for clusters of points in a small area at a fixed pre-determined position~\cite{slingsby2007interactive}. 
These clusters are shown using an enlarged label of the predominant category in the area. While this may lead to many overlapping labels, tag maps have been extended to prevent these overlaps, and adapt the label sizes in more sophisticated ways~\cite{DBLP:conf/apvis/ThomBKWE12, reckziegel2018predominance}. 
However, in tag maps, a point can be unlabeled, far away from its representative label, or even covered by a label of a different category.
A disadvantage of tag maps is that they provide no guarantees on the number of such misrepresented points or the distances to the correct labels.
In this sense, the precision guarantee of point labeling is lost in tag maps.


Another popular method used to aggregate text-based data in information visualisation is word clouds, where typically the objective is to highlight the most frequent keywords in a text~\cite{wordle}. More frequent keywords are visualized using a larger font size. This technique has been extended to work with geo-spatial point data that have keywords associated with each point~\cite{BuchinGeo16}.
These \emph{geo word clouds} place keywords in areas that contain many points with the associated keyword. For each keyword, the associated points are clustered, and for each cluster a good label placement is computed in terms of rotation and absolute position. The labels are scaled based on the number of points in the associated cluster, and are placed in decreasing order of their scaling factor. Whenever the currently placed label overlaps an earlier placed label, the current label is shrunk a little bit, and a new position is determined. The label is placed again, after existing labels with larger scaling factors are considered first. The shrinking ensures that it becomes easier to place the label, albeit further from the intended position.
While such word clouds allow keywords to be displayed in a non-overlapping manner with a large degree of freedom, for example in terms of orientations and absolute placements, they have similar drawbacks as the aforementioned tag maps: there are no guarantees on the precision of such labeling, as the best placement for a label can be blocked by an earlier placed label (see, e.g., Figure~7 in~\cite{BuchinGeo16}).

Therefore, we are facing the challenge to find a hybrid, label-based visualization that simultaneously retains the precision of classical point labeling, where labels are placed exactly at their corresponding points, and produces aggregated word-cloud-like labels that represent spatial patterns of keyword occurrences well. 

\begin{figure}
	\captionsetup[subfigure]{justification=centering}
	\begin{subfigure}[t]{0.47\textwidth}
	    \includegraphics[width=\textwidth]{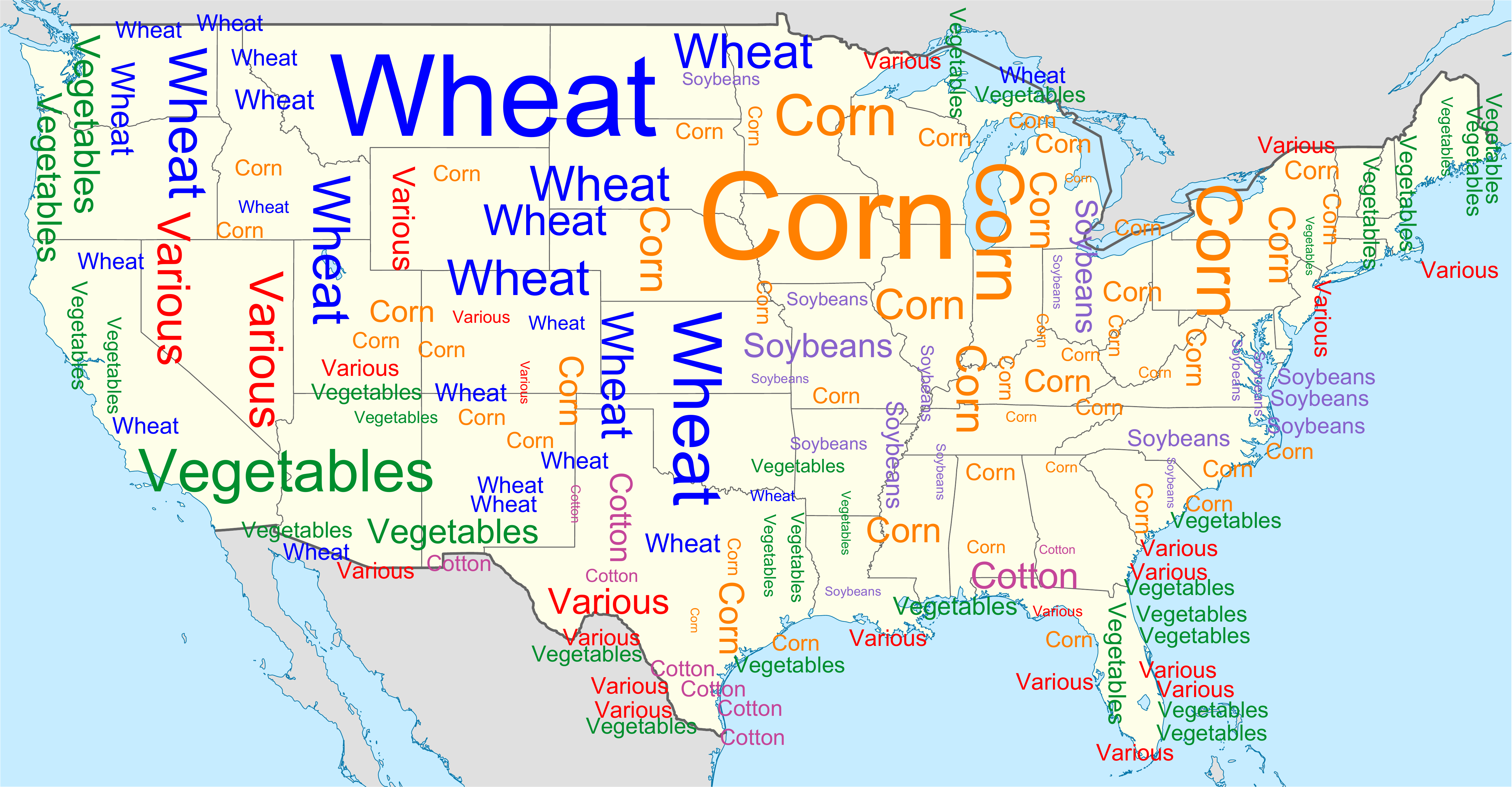}
	    \caption{}
	    \label{fig:teaser:a}
	\end{subfigure}
	\begin{subfigure}[t]{0.47\textwidth}
		\includegraphics[width=\textwidth]{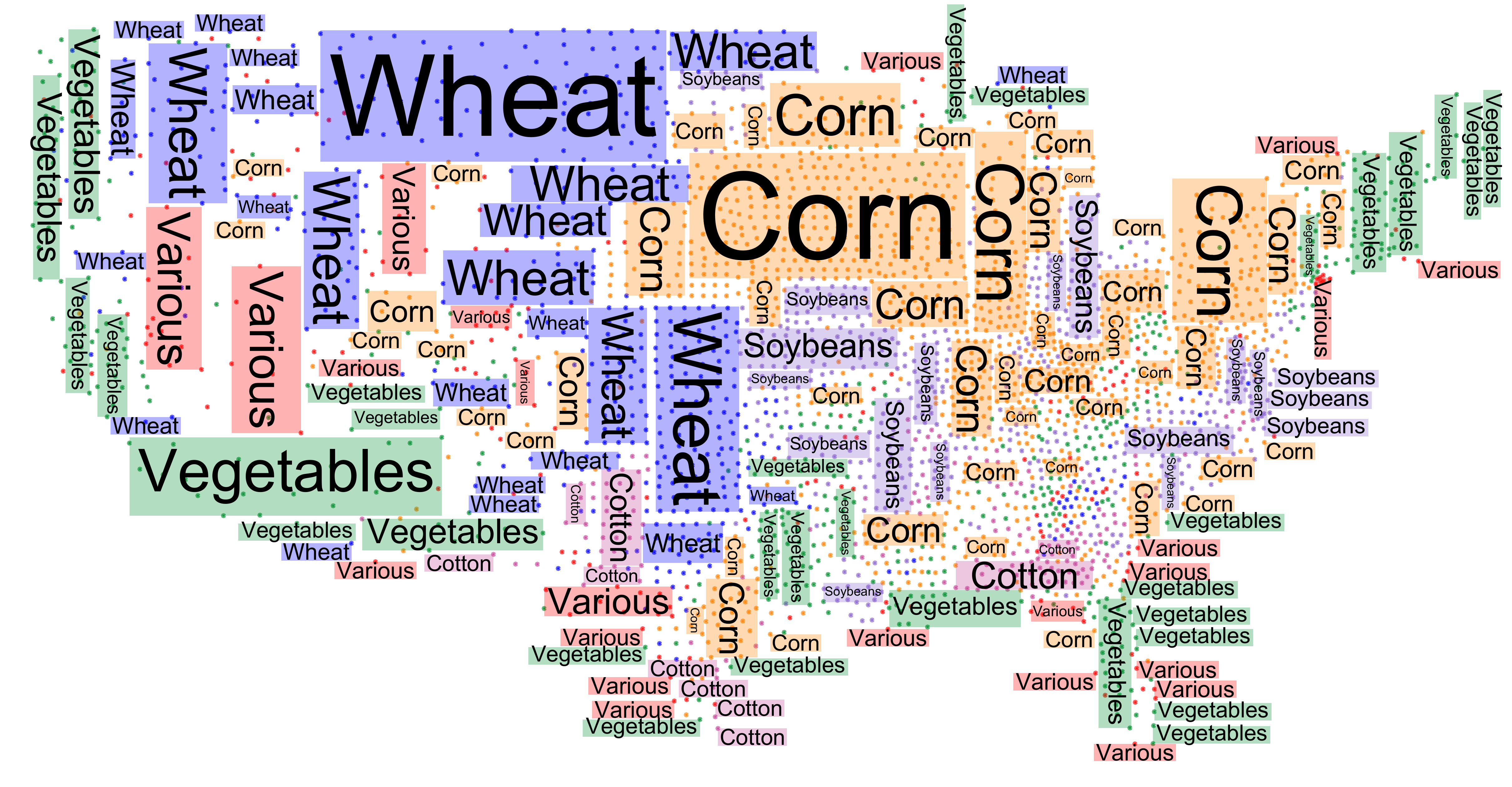}
		\caption{}
		\label{fig:teaser:b}
	\end{subfigure}
	\caption{Worbel visualizations of the dominant crop harvested in each county of the United States. \textsc{\textbf{(a)}}~A Worbel overlay on a map of the US. \textsc{\textbf{(b)}}~The same Worbel drawn over the data points, with rectangles outlining computed area per label.}
	\label{fig:teaser}
\end{figure}

\newpage

\mypar{Contributions} 
In this work, we study a hybrid visualization that combines aspects of word clouds and point labeling, which we call a \emph{Worbel}, see Fig.~\ref{fig:teaser} for an example.
The input consists of a set of points in the plane, which each have a category associated to them. We want to place disjoint textual labels, representing the categories, such that the labels spatially represent the underlying data points of the same category. We develop a formalization of this task as the problem of finding a small set of disjoint category rectangles which cover points of each category while adhering to constraints that bound the amount of \emph{misrepresentation}, the \emph{aspect ratio}, and require a \emph{minimum size}; we call this the \textsc{Rectilinear Point Feature Aggregation} (\RPFA) problem. As our first theoretical contribution, we develop a non-trivial reduction showing that \RPFA\ is \NP-hard and remains so even under severe restrictions.

Since this rules out an efficient exact algorithm for the problem, we turn our attention towards robust heuristics. 
We present a \textsc{MaxSAT}-model for finding an exact solution, as well as a simple but highly effective heuristic algorithm designed to find a sufficiently good (but not necessarily optimal) solution. In a quantitative evaluation we show that this heuristic approach computes solutions that are close in quality to an optimal solution, but requires a much lower computation time than the \textsc{MaxSAT}-model.
Our evaluation includes a case study on real-world data, in which we demonstrate the visual qualities of the Worbel technique.

\mypar{Related Work} 
Map labeling papers~\cite{AgarwalKS98, van1999point, Bhore0N20} often focus on the computation of the maximum independent set of geometric objects as an underlying combinatorial optimization problem. While this turned out to be an appropriate abstraction for point-feature labeling, it omits many labels and can only show a fraction of the input labels due to geometric packing constraints. In our setting, even though we are interested in a labeling that represents the underlying point features well, we provide an aggregated labeling. Specifically, we label multiple points of the same category with a single label and therefore produce a labeling that is closer to an area-labeling~\cite{b-mpeal-01} of implicitly defined areas of homogeneous points.


World clouds have been used as a key tool for text-based visualization, where the goal is to highlight prominent keywords in large amounts of text, emphasizing more important or frequent words by a larger font size; see~\cite{wordle, barth2014semantic, chi2015morphable}.
Moreover, the study of word clouds and tag maps have been extended to metro wordles~\cite{LiDY18} and time varying tag maps~\cite{ReckziegelJ19}.


Besides text-based geographic visualizations, there are other types of thematic maps~\cite{dent1999cartography} to portray spatial patterns of categorical point and area data, e.g., choropleth and chorochromatic maps, dot maps, and generally maps placing symbols and marks that convey data attributes by visual variables such as size, texture, color, or shape, and whose meaning is provided via lookups in a map legend.



The \RPFA\ problem has its roots in  geometric covering problems, and 
many variants of geometric covering problems have been investigated over the years, e.g., $(p,k)$-box covering~\cite{DBLP:journals/comgeo/AhnBDDKKRS11}, class cover problems~\cite{bereg2012class}, and red-blue cover~\cite{AshokKS17, red-blue-cover, DBLP:conf/soda/CarrDKM00}. 
Moreover, geometric covering problems have been widely studied in various algorithmic paradigms such as approximation algorithms (see~\cite{AgarwalP20, ClarksonV07, DBLP:journals/siamcomp/DemaineFHS08}) and parameterized algorithms (see~\cite{PradeeshaPC, AshokKS17}).
Most of these works, however, do not consider the case where the covering shapes must be disjoint. This makes our problem particularly unique, and none of these algorithms can be directly applied in our context.

\mypar{Paper Organization}
In Section~\ref{sec:model}, we introduce \RPFA\ and show that it is \NP-hard. We then present the \textsc{MaxSAT}-model and our greedy heuristic in Section~\ref{sec:algorithms}.
In Section~\ref{sec:evaluation} the results of our quantitative evaluation are presented. Finally we discuss the presented results and potential improvements in Section~\ref{sec:conclusion}.



\section{Theoretical Model}\label{sec:model}





\begin{figure}
    \centering
    \includegraphics[page=1]{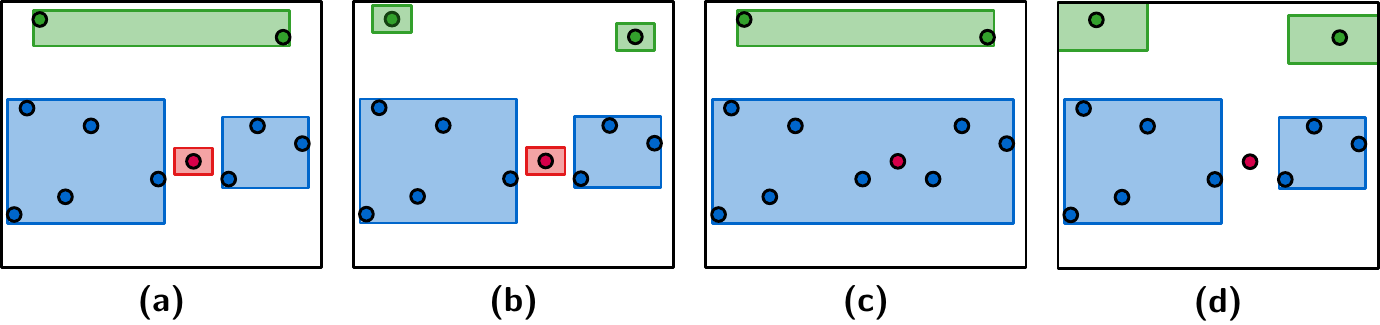}
    \caption{Solutions to \RPFA\ with varying parameters. Distinct labels are indicated as different colors. \textsf{\textbf{(a)}}~Parameters $\rho_l=0$, $\rho_u=\infty$, and $t=\rho_t=f=0$. \textsf{\textbf{(b)}}~Lowered~$\rho_u$ w.r.t.~\textsf{\textbf{(a)}}. \textsf{\textbf{(d)}}~Increased $t$ to 1 w.r.t.~\textsf{\textbf{(a)}}. \textsf{\textbf{(d)}}~Increased $f$ w.r.t.~\textsf{\textbf{(a)}}.}
        \label{fig:problem-parameters}
\end{figure}

We begin by providing a formalization of the problem that underlies the task of aggregating point features in a Worbel.

\mypar{Formal Problem Description} In \textsc{Rectilinear Point Feature Aggregation} (\RPFA), we are given a set~$P$ of points in the plane $\mathbb R^2$, a set~$L$ of different labels, 
and a label function $\ell \colon P\rightarrow L$ which assigns each point in~$P$ to a label in~$L$. 
The task is to cover all points in~$P$ with a minimum-cardinality set~$S$ of pairwise-disjoint axis-aligned rectangles, where each rectangle $R\in S$ is also assigned a label $l_R\in L$. In a Worbel, label $l_R$ will be printed inside each $R$, scaled to the size of $R$ and rotated to match the major axis. Each rectangle $R\in S$ must additionally adhere to the following constraints.

\begin{description}
 \item[Aspect ratio] The aspect ratio~$a_R$ of $R$ should be close to the aspect ratio $a_{l_R}$ of the bounding box of label~$l_R$.  We introduce parameters $\rho_l,\rho_u$, with $0 \leq $ $\rho_l < 1 < \rho_u$ as the lower and upper bounds on the ratio between $a_R$ and $a_{l_R}$. Specifically, we ensure that $\rho_l \leq \frac{a_R}{a_{l_R}} \leq \rho_u$ for each rectangle $R\in S$.
 \item[Tolerance for misrepresentation] The points covered by $R$ should ``predominantly'' be assigned to the label~$l_R$ of $R$. To this end, we introduce a tolerance threshold~$t$ and a tolerance ratio~$\rho_t$, which quantify how many points in~$R$ can be of a different label. The number~$|R|$ of points in~$R$ that have a label different from~$l_R$ is at most $\min(t, \rho_t \cdot |R|)$.
 \item[Minimum font size] We enforce a minimum font size~$f$ to ensure that labels remain readable on screen, and cannot be scaled down arbitrarily. To this end, the minor axis of~$R$ is not allowed to be smaller than~$f$.  
\end{description}

Observe that it may not be possible to cover all points by disjoint rectangles due to the constraint on the minimum font size; in these cases we simply aim to cover as many points as possible. we say that a set $S$ of disjoint rectangles is \emph{viable} if it covers a maximum-cardinality subset of $P$
while satisfying all the constraints specified above. Hence, in its full generality an instance of \RPFA\ is a tuple $(P,L,\ell,\rho_l,\rho_u,t,\rho_t,f)$, and an \emph{optimal solution} is a minimum-cardinality viable set $S$ of rectangles. For complexity-theoretic purposes, let \RPFA$^*$ be the decision of version of \RPFA; there, we are additionally given a bound $k\in \mathbb{N}$ and ask whether there exists a set $S$ of viable rectangles of cardinality at most $k$.


\mypar{Problem Complexity}
It is easy to observe that \RPFA$^*$ is in \NP. As one of our main theoretical contributions, we will show that \RPFA$^*$ is not only \NP-hard, but it remains \NP-hard even when restricted to a very specific subset of inputs. In particular, let \RPFA$_0$ be the restriction of \RPFA$^*$ to inputs where $L=\{$red, blue$\}$, $\rho_l=0$, $\rho_u=\infty$, and $t=\rho_t=f=0$. Intuitively, \RPFA$_0$ is the restriction of \RPFA$^*$ to the subcase where all constraints are disregarded and we have only two labels.

\begin{theorem}
\label{thm:hard}
\RPFA$_0$ is \NP-hard.
\end{theorem}

To prove Theorem~\ref{thm:hard}, we reduce from a problem called \textsc{Disjoint Box Covering in a Rectilinear Polygon} (\DBCR). In \DBCR, we are given a rectilinear polygon $\psi$ (possibly with holes), a set $\alpha$ of $n$ points inside $\psi$ (which we call \emph{elements} for disambiguation), and a bound $\omega$.

The question is whether there exists a set $\mathcal{Q}$ of at most $\omega$ pairwise disjoint axis-aligned rectangles such that each element in $\alpha$ is contained in (i.e., \emph{covered by}) a rectangle in $\mathcal{Q}$ and each rectangle is fully contained inside $\psi$. 

Two reflex corners $x,y$ of $\psi$ are called \emph{opposite} if the unique axis-aligned rectangle $U$ which has $x$ and $y$ as its corner points is fully contained in $\psi$ and only intersects the boundary of $\psi$ in $x$ and $y$; in this case, we call $U$ the $x$-$y$-\emph{spanning rectangle}. \DBCR\ was recently shown to be \NP-hard, even when restricted to instances with the following properties~\cite[Theorem 2.3]{blnw-dcrp-21}.
\begin{itemize}
\item each $x$-$y$-spanning rectangle contains at least one element, 
\item each corner of $\psi$ and element lie on a finite integral \emph{embedding grid} whose size is polynomial in the input size, and 
\item each pair of two elements and also each pair of one element and one corner of $\psi$ have distinct $x$- and $y$- coordinates (i.e., they lie in general position with respect to each other).
\end{itemize}

\begin{proof}[Proof of Theorem~\ref{thm:hard}]
We provide a polynomial reduction from \DBCR\ to \RPFA$_0$. Intuitively, the main difference between the two problems is that while in \DBCR\ there is a clearly defined bounding polygon $\psi$, in \RPFA$_0$ the placement of rectangles is restricted exclusively by the presence of points with a different label (which themselves need to be covered by a rectangle of the other color). To overcome this difficulty, the reduction will use a construction that simulates the presence of $\psi$; this is also where the property of spanning rectangles mentioned earlier will come into play.
\begin{figure}
    \centering
    \includegraphics{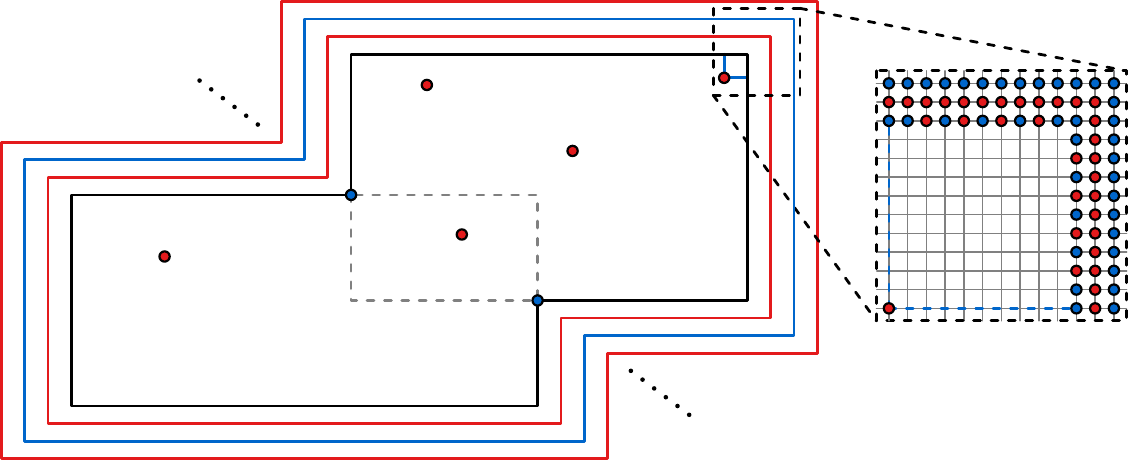}
    \caption{Construction of $\III$: the initial grid is refined by a factor 10, there are alternating layers of red and blue points around initial polygon $\psi$, and the boundary consists of (mostly) alternating red and blue points. Note that each $x$-$y$-spanning rectangle will contain a red point.}
    \label{fig:reduction}
\end{figure}
The reduction will take as input an instance $\III'=(\psi',\alpha',\omega)$ of \DBCR\ and construct an instance $\III$ of \RPFA$_0$. Figure~\ref{fig:reduction} shows illustrates one such instance $\III$. We begin by refining the resolution of the embedding grid by a factor of $10$; the sole purpose of this is to ensure that the coordinates of each element and corner differ by at least $10$. We will call the points of the embedding grid that lie on the boundary of $\psi$ \emph{boundary points}. We iteratively partition all the points 
outside of $\psi$, but inside the embedding grid, into layers as follows. The first layer contains all points which are adjacent to a boundary point in the 8-point neighborhood topology (i.e., these are either axis-adjacent or diagonally adjacent to a boundary point). Once we have identified the set of layer-$i$ points, we define the set of layer-$(i+1)$ points by repeating this procedure---in particular, a point is in layer $i+1$ if it does not belong to any layer up to $i$ and is adjacent to a point in layer $i$ in the 8-point neighborhood topology.

We are now ready to construct the instance $\III$ of \RPFA$_0$, which will be the output of our reduction. First, we add into $P$ all points in odd layers as blue points, and all points in even layers as red points. Observe that this creates a sequence of alternating red and blue axis-aligned polygonal contours around 
$\psi$, see Fig.~\ref{fig:reduction}. Let $\delta$ be the minimum number of straight-line segments required to cover all edges of these contours, and notice that $\delta$ can be computed in polynomial time via a straightforward greedy procedure. 

Each element in $\alpha$ will also be added to $P$ as a red point. Intuitively, we now have an instance of \RPFA$_0$ where we can ascertain precisely how many rectangles are needed to cover the points in the layers as long as the points in layer $1$ are not allowed to be covered by rectangles that intersect $\psi$. Moreover, under the same assumption about the rectangles covering the points in layer $1$, the internal area of $\psi$ could be used exclusively for red rectangles covering $\alpha$, precisely matching the behavior of the original instance $\III'$. Unfortunately, there is currently nothing that would prevent the blue rectangles covering the boundary points from entering $\psi$; that is where the next, crucial step of the reduction comes in.

We say that a boundary point is \emph{important} if it is either a corner of $\psi$, or if it has the same $x$- or $y$-coordinate as a point in $\alpha$. All important boundary points will be added to $P$ as blue.
All remaining boundary points will also be added to $P$ in a way which ensures that every two such points that lie ``opposite'' to each other on the boundary of $\psi$ receive a different color. In particular, each non-important boundary point that lies on:
\begin{itemize}
\item a bottom boundary of $\psi$ is red $\iff$ its $x$-coordinate is odd;
\item a top boundary of $\psi$ is blue $\iff$ its $x$-coordinate is odd;
\item a left boundary of $\psi$ is red $\iff$ its $y$-coordinate is odd;
\item a right boundary of $\psi$ is blue $\iff$ its $y$-coordinate is odd.
\end{itemize}

The result is that the boundary points form a checkered red-blue pattern, but with some local irregularities around the important points. Finally, we set $k$ to be equal to the sum of (a) $\delta$ (b) $\omega$, and (c) the total number $\sigma$ of color switches that occur on the boundary points when performing a walk along each boundary segment of $\psi$.


Clearly, the construction described above can be carried out in polynomial time, and hence all that remains is to argue correctness. For the forward direction, assume that $\III'$ was a \texttt{YES}-instance of \DBCR\ and admits a solution $\mathcal{Q}$ (i.e., a set of rectangles covering $\alpha$, contained in $\psi$, and of size at most $\omega$). We can then construct a solution $S$ for $\III$ by (a) covering all the points in the layers using $\delta$-many rectangles, (b) having one red rectangle copy the placement of each of the at most $\omega$-many rectangles in $\mathcal{Q}$, and (c) using one appropriately-colored rectangle to cover each consecutive set of red or blue boundary points that occur on a walk around each boundary segment of $\psi$ ($\sigma$-many in total). It is easy to verify that this construction produces a set $S$ of at most $k$ rectangles which cover all the points in $P$.

On the other hand, assume that $\III$ admits a solution $S$. First of all, we will argue that $S$ can be assumed to use $\delta$-many rectangles to cover all points that lie in the individual layers without interfering with $\psi$. For each layer other than the first, each point can only lie in a rectangle with other points in the same layer which makes the claim obvious, but a careful argument is needed for the first layer. 

To this end, assume that $S$ covers some blue point $p$ in layer $1$ with a rectangle $T$ that also contains a blue point $q$ that does not lie in layer $1$; this point would have to be a boundary point. But by the coloring used for the boundary points, $T$ can then only contain at most two other blue boundary points adjacent to $q$ and at most two other blue points in layer $1$ adjacent to $p$. Let $a$ and $b$ be the two unique blue points on layer $1$ that are adjacent to $T$ when walking along layer $1$, and notice that the grid refinement we performed at the very beginning prevents $a$ and $b$ from both being corners simultaneously; without loss of generality, let us assume that $b$ is not a corner. Then the two axis-neighboring points of $b$ outside of layer $1$ can be assumed to be red; if they are not, we can add them to $T$ along with their axis-neighbors on the boundary without affecting the validity of $S$. Let $B$ be the rectangle in $S$ which contains $b$, and notice that $B$ can only contain points in layer $1$. We can now alter $S$ as follows: extend $B$ to cover all blue points in $T$, and restrict $T$ only to those points which lie outside of layer $1$. 


We have now obtained a new solution which has equal size as the original $S$, but contains one less occurrence of a rectangle containing points in layer $1$ and also on the boundary. By iterating this procedure, we arrive at an equivalent solution $S$ such that no rectangle containing points in layers intersects $\psi$; these rectangles hence form walks along the individual layers and $S$ contains at least $\delta$-many of them.

Now, consider the rectangles in $S$ which contain at least one boundary point. Crucially, by combining
\begin{itemize}
\item the construction of the colors on the boundaries, \item the fact that each spanning rectangle in $\III'$ contains at least one point and all boundary corner points are blue, and 
\item the fact that all corners lie in general position,
\end{itemize}
we obtain that every rectangle in $S$, containing at least one boundary point, can only contain other boundary points, and that these must lie in a consecutive strip of points along the boundary; in other words, we may assume without loss of generality that such a rectangle is either a line segment which follows the boundary, or a $2\times 2$ square containing one blue corner point along with two adjacent blue points\footnote{A careful reader may notice that this degenerate case may only occur for points in the top-right corner, but such rectangles do not intersect the space required for covering the red points in $\psi$ due to the fact that each red point has a gap of at least $10$ grid points from every corner point due to the grid refinement used.}.

As a consequence, we obtain that all rectangles in $S$ can be partitioned into three pairwise-disjoint sets: those used to cover the points in the layers, those used to cover the boundary points, and those used to cover the points inside $\psi$. Since $|S|\leq k$ and the first two subsets of $S$ have a cardinality of at least $\delta$ and $\sigma$, we obtain that the number of rectangles in the third set is at most $\omega$. But then the set of rectangles used by $S$ to cover the points inside $\psi$ is a solution for $\III'$, as required. \qedhere
\end{proof}

In view of Theorem~\ref{thm:hard}, we cannot hope for an exact algorithm for \RPFA\ with a polynomial-time runtime guarantee. That is why we turn to efficient heuristics for the problem, and show that these can produce good Worbels in many scenarios of interest.
       
\section{Algorithms}\label{sec:algorithms}
We introduce a greedy heuristic and an exact \textsf{SAT} model for solving \RPFA. Both methods take a set of candidate rectangles as input, and compute a viable set of disjoint rectangles among the candidates. We first discuss how to find an appropriate set of candidate rectangles, given the input parameters, aspect ratio bounds~$\rho_l, \rho_u$, tolerance~$t,\rho_t$, and minimum font size~$f$.

\subsection{Computing Candidate Rectangles}
\label{subsec: candidates set}
A naive idea for computing the candidate set~$\mathcal{R}$ of rectangles
would be to simply consider the set of all rectangles whose left, right, bottom and top boundary touches some point in $P$. This would result in $\Theta(|P|^4)$ combinatorially different rectangles, from which a subset $S$ of at most $|P|$ disjoint rectangles to cover $P$ will be selected.


Since we want the candidate rectangles to adhere to the input parameters, we often need to consider far less rectangles. Hence we instead consider each pair of distinct points $p_i,p_j\in P$, and compute a set~$\mathcal{R}_{i,j}$ of candidate rectangles as follows. We first compute the smallest bounding box~$R_{\text{base}}$ of $p_i$ and $p_j$. This rectangle will be the base for a set of rectangles that extends leftward and rightward, inside the horizontal strip defined by $p_i$ and $p_j$ (see Figure~\ref{fig:candidate-strip}).

We first check which label $l\in L$ is the predominant label among the points in~$R_{\text{base}}$, and we additionally make sure that the number of points with a label different from $l$ does not exceed $\min(t, \rho_t \cdot |R|)$.
If $R_{\text{base}}$ fails the tolerance check, 
we move to the next pair $p_i,p_j\in P$, otherwise we assign $l_{R_{\text{base}}} \coloneqq l$ and proceed.

\begin{figure}
    \centering
    \includegraphics{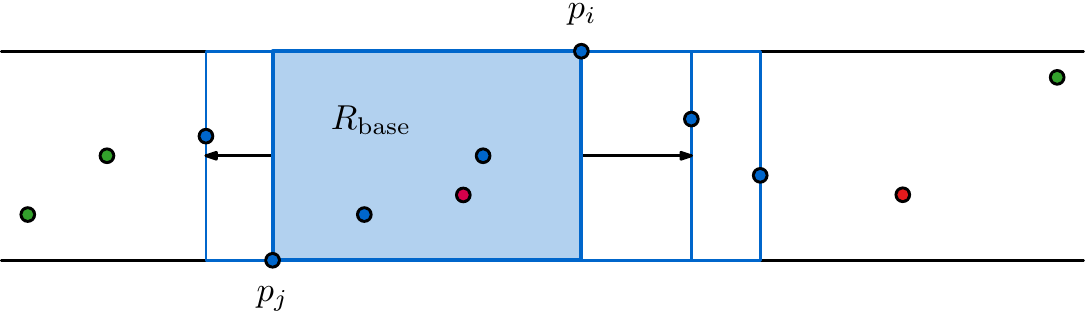}
    \caption{Horizontal strip defined by $p_i$ and $p_j$, in which $R_{\textnormal{base}}$ will be extended leftward and rightward, until one of the input constraints (here misrepresentation) is violated.}
    \label{fig:candidate-strip}
\end{figure}

Rectangle $R_{\text{base}}$ is now extended leftwards, creating a new candidate for each additional point covered by an extended rectangle. Observe that this operation has the following consequences for the aspect ratio, tolerance and minimum font size. The aspect ratio of vertically oriented rectangles will increase, until they become squares and the orientation becomes horizontal. From that point on, the aspect ratio will decrease, and we stop extending when the aspect ratio of an extended rectangle $R$, with label~$l_R$ and label aspect ratio $a_{l_R}$ falls below $\rho_l \cdot a_{l_R}$. The number~$m$ of points with a label different from $l_R$ may both increase or decrease, since newly covered points may be of label~$l_R$ or of a different label, respectively. However we stop extending as soon as $m$ exceeds $\min(t, \rho_t \cdot |R|)$. The font size will monotonically increase, and therefore does not influence when we stop extending. Thus for all extended rectangles we check the aspect ratio and tolerance, to see if we need to stop extending. Rectangle~$R_{\text{base}}$ along with the extended rectangles will form the set $\mathcal{R}_{\text{left}}$.
\begin{figure}[b]
\centering
	\captionsetup[subfigure]{justification=centering}
	\begin{subfigure}[t]{0.4\textwidth}
	\centering
		\includegraphics[page=1]{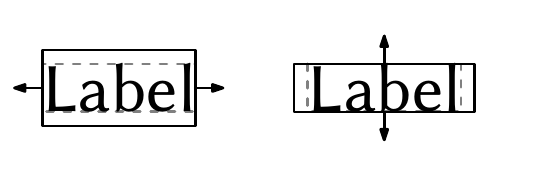}
		\caption{}
	\end{subfigure}
	\begin{subfigure}[t]{0.55\textwidth}
	\centering
	    \includegraphics[page=2]{figs_arxiv/extend-rectangle.pdf}
	    \caption{}
    \end{subfigure}
    \caption{Stretching $R$ closer to the aspect ratio of the corresponding label~$l_R$. \textsf{\textbf{(a)}} Stretching along major/minor axis depends on aspect ratio of $R$ compared to $l_R$. \textsf{\textbf{(b)}} Three candidates for stretching: upwards, downwards, or both.}
    \label{fig:extend-rectangle}
\end{figure}
We then do the same extension operation on all rectangles in $\mathcal{R}_{\text{left}}$, but this time we extend rightwards. We have the same stopping criteria, checking for aspect ratio and tolerance of the extended rectangles. The newly extended rectangles, together with the rectangles in $\mathcal{R}_{\text{left}}$, form the set $\mathcal{R}_{\text{ext}}$.

We check for each rectangle $R\in \mathcal{R}_{\text{ext}}$, with label~$l_R$ and label aspect ratio~$a_{l_R}$, whether the aspect ratio~$a_R$ is close enough to the aspect ratio~$a_{l_R}$, by ensuring that $\rho_l \leq \frac{a_R}{a_{l_R}} \leq \rho_u$. If this is not the case, then we try to stretch $R$ such that its aspect ratio gets closer to $a_{l_R}$, without covering any additional points. We do so by increasing the side length along either the major or minor axis (see Figure~\ref{fig:extend-rectangle}a). For example, assume the major axis of $R$ is horizontal and $a_R < \rho_l \cdot a_{l_R}$. A similar procedure can be used when $R$ is vertical or $a_R > \rho_u \cdot a_{l_R}$. We stretch $R$ along its minor axis, creating up to three new candidates $\{R^u, R^d, R^c\}$: $R^u$ stretches only upwards, $R^d$ stretches only downwards, and $R^c$ uses space above and below $R$ (see Figure~\ref{fig:extend-rectangle}b). We stretch $R$ until $a_R = \rho_l \cdot a_{l_R}$, or until we hit another point $p\not\in R$; we never introduce new points into $R$ since the resulting candidates will be handled by a different choice of $p_i,p_j$ anyway. (If $R$ is vertical those candidates would already be in~$\mathcal{R}_{\text{ext}}$.)

We try to center $R^c$ around $R$, using equal space above and below, but a point $p\not\in R$ may be located nearby and this may force us to make adjustments. For example, if $p$ is located above $R$, then we use all the available space above $R$, and as much as necessary below $R$, to ensure $a_R = \rho_l \cdot a_{l_R}$. However, if there is also a point too close below $R$, we cannot reach $a_R = \rho_l \cdot a_{l_R}$ without covering additional points. We remove $R$ from $\mathcal{R}_{\text{ext}}$ and add the subset of $\{R^u, R^d, R^c\}$ of rectangles which pass the aspect ratio check and do not cover any additional points.

Next, we check for each rectangle $R\in\mathcal{R}_{\text{ext}}$ whether label~$l_R$ can be placed inside $R$ at the minimum font size~$f$. Remember that we always align the major and minor axis of a rectangle and its assigned label. We remove $R$ from $\mathcal{R}_{\text{ext}}$ if this check is not passed. The remaining rectangles form the candidate set~$\mathcal{R}_{i,j}$.

Notice that a candidate set~$\mathcal{R}_{i,j}$ does not produce any rectangles covering single points. To this end, we add candidates $\mathcal{R}_i$ for each point $p_i\in P$ by using the aforementioned stretching operation both horizontally and vertically from $p_i$. Since we created three options when stretching in one direction, there are nine possible outcomes when stretching both horizontally and vertically. The union of all candidate sets $\mathcal{R}_{i,j}$ and $\mathcal{R}_i$ forms the candidate set~$\mathcal{R}$ that we use in our greedy heuristic and the \textsf{SAT} model.
$$\mathcal{R} \coloneqq \bigcup_{\{p_i,p_j\}\subseteq P} \mathcal{R}_{i,j} \cup \bigcup_{p_i\in P} \mathcal{R}_i$$

Given our candidate set~$\mathcal{R}$ of rectangles that satisfy our constraints, we want to find a smallest subset $S\subseteq \mathcal{R}$ of disjoint rectangles that covers all points in $P$ or---if the minimum font size prevents some points from being covered---as many points in $P$ as possible. Additionally, each candidate set $\mathcal{R}_{i,j}$ does not necessarily contain all viable rectangles, since we stop whenever the tolerance restriction is not met. In case we do not find all viable rectangles, even the smallest subset $S\subseteq\mathcal{R}$ will not be an optimal solution for \RPFA, but it is the best solution we can find heuristically.



\subsection{Weighted Independent Set Model}
\label{subse:wmis_model}
In this section, we show how to find a good set of disjoint rectangles in the candidate set~$\mathcal{R}$, by modeling our problem as the \textsc{Weighted Independent Set} (\textsc{WIS}) problem in the intersection graph on~$\mathcal{R}$.
A benefit of this approach is that it allows us to utilize algorithmic approaches developed for this well-studied problem in previous works~\cite{ChalermsookW21, AdamaszekW13, GKMMPW-4appx}.
In our setting, \textsc{WIS} takes as input a set of rectangles $\mathcal{R}$, and weight function $w: \mathcal{R} \rightarrow \mathbb{R^+}$ which assigns a positive weight to each rectangle in $R$. We want to find a set $S\subseteq \mathcal{R}$, such that the rectangles in $S$ are pairwise disjoint and their sum of weights is maximized.

Given an instance~$\III$ of \RPFA\ and a candidate set $\mathcal{R}$ consisting of all viable rectangles in the instance $\III$, we construct an instance~$\III'$ of \textsc{WIS} by assigning weights to all candidate rectangles in~$\mathcal{R}$ as follows. 
Let $R$ be a rectangle in the candidate set, and $|R|$ denote the number of points that lie inside $R$.   
Furthermore, let $n = |P|$ be the number of points in $\III$.
For the instance $\III'$ we take the candidate rectangles~$\mathcal{R}$, and assign to each rectangle $R\in \mathcal{R}$ the weight $2n|R|-1$.
We say that $\III'$ is the corresponding \textsc{WIS} instance of $\III$.

Let $S' \subset \mathcal{R}$ be a valid independent set, i.e., the set $S'$ consists of pairwise non-overlapping rectangles.
Then, in the instance $\III'$, the weight of the set $S'$ is $W(S') = \sum_{R \in S'} (2n|R|-1)$.
Let $P_{S'}$ be the point set consisting of points in $P$ that are covered by rectangles in $S'$.
Since all rectangles in $S'$ are disjoint, we can rewrite the above equality to get $W(S') = 2n|P_{S'}|-|S'|$.

We can now prove that, for this choice of weights, a maximum-weight independent set in $\III'$ corresponds to a
set of disjoint rectangles in $\III$ which (1) maximizes the number of covered points, and then among all such sets (2) is of minimum size.

\begin{lemma}
\label{lem:wmis}
Let $S,S'\subseteq \mathcal{R}$ be such that $S$ covers more points than $S'$, i.e., $|P_S| > |P_{S'}|$. Then $S$ has a larger weight than $S'$ in the \textup{\textsc{WIS}} instance, i.e., $W(S) > W(S')$.
\end{lemma}
\begin{proof}
Since $S$ and $S'$ both consist of pairwise non-overlapping rectangles, each point in $P$ can be covered by at most one rectangle in $S$ and also in $S'$.
Then by the pigeon hole principle, $\max(|S|,|S'|) \leq n$.
Thus:
\begin{align*}
    W(S) = 2n|P_S| -|S| &> 2n(|P_S|-1) \\
    &\geq 2n|P_{S'}| \geq 2n|P_{S'}| - |S'| =W(S'). \qedhere
\end{align*}

\end{proof}

\begin{theorem}
\label{thm:wmis_model_1}
Let $\III$ be an instance of \RPFA, and $\III'$ be the corresponding \textsc{\textup{WIS}} instance.
Then an optimal solution~$S^*$ of $\III$ corresponds to a maximum weight independent set in $\III'$. 
\end{theorem}
\begin{proof}
We make a case distinction on or not whether $\III$ has a non-zero minimum font size~$f$.

\textbf{Minimum font size~$f=0$:}
In this case, $S^*$ must cover all points in $P$. 
That means $P_{S^*} = P$, and implies $W(S^*) = 2n|P|-|S^*| = 2n^2-|S^*|$.
Now consider the set $S$ which corresponds to $S^*$, First, observe that $S^*$ corresponds to an independent set by definition, and hence we only have to argue that $S^*$ corresponds to an independent set that maximizes the weight. 
Let $S$ be an independent set in $\III'$. 
If $S$ does not cover all points, then by Lemma~\ref{lem:wmis} the weight of $S$ cannot be larger than $W(S^*)$. 
Thus assume that $S$ covers all points in $P$, and assume for contradiction that $W(S)>W(S^*)$. We get
$$2n^2-|S| = W(S) > W(S^*) = 2n^2-|S^*|.$$ Hence it must hold that $|S| < |S^*|$. This contradicts the fact that $S^*$ is an optimal solution for $\III$, since $S$ also covers all points, is independent, and consists of fewer rectangles.

\textbf{Minimum font size~$f>0$:}
In this case, an optimal solution~$S^*$ need not cover all points in $P$. However, $S^*$ covers as many points as possible.
That is, for any independent set $S$ in $\mathcal{R}$, we have $|P_S| \leq |P_S^*|$.
If $S$ covers fewer points than $S^*$, then by Lemma~\ref{lem:wmis}, the weight of $S$ is smaller than the weight of~$S^*$.
Since $S^*$ is an independent set by definition, there must be an independent set $S$ that covers at least as many points as~$S^*$. We can therefore repeat the argument in the previous case, showing that there cannot be a solution $S$ for $\III'$ that has higher weight than $S^*$, and thus fewer rectangles.
\end{proof}

\begin{lemma}
\label{thm:wmis_model_2}
Let $\III$ be an instance of \RPFA, and $\III'$ be the corresponding \textup{\textsc{WIS}} instance.
Then each maximum weight independent set in instance $\III'$ is an optimal solution~$S^*$ in $\III$.
\end{lemma}
\begin{proof}
Consider a maximum independent set $S$ in the instance $\III'$. 
By Lemma~\ref{lem:wmis}, we know that $S$ must cover as many points as possible.
Furthermore, there cannot be an independent set $S'$ in $\III'$, that covers as many points as $S$ ($|P_S|=|P_{S'}|$), but has fewer rectangles ($|S|>|S'|$). For such an independent set $S'$ the inequality
$W(S') = 2n|P_{S'}| - |S| > 2n|P_S| - |S| = W(S)$ would hold, contradicting the optimality of $S$. 
\end{proof}


\subsection{Exact Solver Based on \textsc{MaxSAT}}
\label{subsec: exact}
Recently, Klute et al.~\cite{kllns-esl-19} proposed an exact \textsc{MaxSAT} model for the point feature labeling problem. 
They achieved reasonable running time for real-world cartographic applications.
Using a similar idea, we propose an exact approach based on a \textsf{SAT} model, modeling weighted independent set as a set of clauses and Boolean variables. 

\mypar{Satisfiability}
A clause $c$ is a disjunction of a set of Boolean variables, where each variable~$v$ appears either as a positive literal~$v$ or a negative literal~$\neg v$ in $c$.
Given a variable set~$V$, a truth assignment $\alpha: V\rightarrow \{\texttt{true}, \texttt{false}\}$ maps a truth value to each variable in $V$.  Given a truth assignment, one evaluates a clause $c$ as true if and only if at least one literal in $c$ is true.
A \emph{weighted} \textsf{SAT} formula $\phi$ in conjunctive normal form is a conjunction of a set of clauses $\phi = c_1 \land \dots \land c_m$, where each clause $c_i$ is assigned a positive weight $w(c_i)$. Given a weighted \textsf{SAT} formula $\phi$, the weighted maximum satisfiability problem (Weighted \textsc{MaxSAT}) asks for a truth assignment $\alpha$ that maximizes the total weight of the satisfied clauses in~$\phi$.


We now encode an instance $\III$ of the \textsc{WIS} problem, consisting of the set $\mathcal{R}$ of weighted candidate rectangles,
in a weighted \textsf{SAT} formula $\phi_\III$ as follows. For each rectangle~$R\in\mathcal{R}$, we introduce a variable $v_R$.
We then construct two types of clauses. 
For each rectangle~$R$, we introduce a rectangle clause $c_R = v_R$ consisting of the single positive literal~$v_R$. 
The weight of $c_R$ is the weight of the corresponding rectangle $R$, namely $w(c_R) = 2n|R|-1$. 
For each intersection between two rectangles $R$ and $R'$, we introduce an intersection clause $c_{R,R'} = \neg v_R \lor \neg v_{R'}$. 
Each intersection clause $c_{R,R'}$ has weight $w(c_{R,R'}) = 2|\mathcal{R}|n^2$.

Given a truth assignment $\alpha$, the set of true variables corresponds to a set $S_{\alpha}\subseteq \mathcal{R}$ of rectangles. 
Conversely, given a set $S\subseteq\mathcal{R}$, the corresponding truth assignment $\alpha_S$ assigns \texttt{true} to a variable~$v_R$ if and only if $R\in S$. 

\begin{restatable}{lemma}{hardCLAUSE}
\label{lem:hard-clause-independent}
A set $S\subseteq \mathcal{R}$ is an independent set if and only if its corresponding truth assignment $\alpha_S$ satisfies all intersection clauses. 
\end{restatable}
\begin{proof}
For every pair of intersecting rectangles $R$ and $R'$ in the \textsc{WIS} instance $\III$, we introduced an intersection clause $c_{R,R'} = \neg v_R \lor \neg v_{R'}$ in $\phi_\III$. 
Given a set $S\subseteq\mathcal{R}$ and its 
truth assignment $\alpha_S$, $S$ contains $R$ and $R'$ simultaneously if and only if $\alpha_S$ sets both $v_R$ and $v_{R'}$ to \texttt{true}, i.e., if and only if clause $c_{R,R'}$ evaluates to \texttt{false}.
\end{proof}

Given a truth assignment $\alpha_\text{max}$ that maximizes the total weight of satisfied clauses in $\phi_\III$, we get the following statements.
\begin{restatable}{lemma}{hardCLAUSEmax}
\label{lem:hard-clause-max}
Every intersection clause is satisfied in $\alpha_\text{max}$. 
\end{restatable}
\begin{proof}
Given a rectangle $R$, its rectangle clause $c_R$ has weight $2n|R|-1 < 2n^2$.
Thus the total weight of all rectangle clauses is bounded by the weight $2|\mathcal{R}|n^2$ of a single intersection clause. 
Consider the special truth assignment $\alpha_{\texttt{false}}$ that assigns \texttt{false} to each variable.
Using this special assignment $\alpha_{\texttt{false}}$, each intersection clause is satisfied. 
Now suppose, for contradiction, that there exists an intersection clause that evaluates to \texttt{false} by $\alpha_{\text{max}}$. 
Then the total weight of true clauses in the assignment $\alpha_{\text{max}}$ is less than the total weight of $\alpha_{\texttt{false}}$, even in the case that $\alpha_{\text{max}}$ satisfies all rectangle clauses. This contradicts the assumption that $\alpha_{\text{max}}$ maximizes the total weight.
\end{proof}

\begin{theorem}
Let $\alpha$ be a truth assignment of $\phi_\III$ and $S_{\alpha}$ be the rectangle set corresponding to all variables that are assigned \textup{\texttt{true}} in $\alpha$.
Assignment~$\alpha$ maximizes the total weight of satisfied clauses if and only if the set~$S_{\alpha}$ is a maximum weight independent set in $\mathcal{R}$. 
\end{theorem}
\begin{proof}
Suppose $S_{\alpha}$ is a maximum-weight independent set, then by Lemma~\ref{lem:hard-clause-independent} its corresponding assignment $\alpha$ must satisfy all intersection clauses. By optimality of $S_{\alpha}$, the total weight of satisfied rectangle clauses is also maximized by $\alpha$, and hence $\alpha$ maximizes the overall weight of satisfies clauses.

Conversely, assume that $\alpha$ is an assignment that maximizes the weight of all satisfied clauses. The corresponding rectangle set $S_\alpha$ is an independent set, by Lemmas~\ref{lem:hard-clause-independent} and~\ref{lem:hard-clause-max}.  
Furthermore, since $\alpha$ satisfies all intersection clauses and maximizes the total weight of satisfied rectangle clauses, $S_\alpha$ has the maximum weight among all independent sets.
\end{proof}

After we encode our problem as a weighted \textsc{MaxSAT} problem, we solve it using the open-source \textsc{MaxSAT} solver MaxHS~3.0~\cite{maxhs}.


\subsection{Greedy Heuristic}
\label{subsec: greedy}
Since the running time of the exact \textsc{MaxSAT} solver is infeasible for larger instances (see Section~\ref{sec:evaluation}), we propose a simple heuristic to find solutions for \textsc{WIS} more quickly. We first assign weights to each rectangle in $\mathcal{R}$, as explained in Section~\ref{subse:wmis_model}, to ensure that an optimal solution for \textsc{WIS} corresponds to an optimal solution for \RPFA. The heuristic starts by choosing the rectangle in $\mathcal{R}$ that has the highest weight as the initial solution~$S$. We then consider the remaining rectangles in order of descending weight, and greedily try to add them to $S$, one by one. A rectangle~$R\in\mathcal{R}$ can either overlap with a rectangle in~$S$, or it is completely disjoint from all rectangles in~$S$. In the former case, we discard~$R$, while in the latter case, $R$ is added to $S$. We proceed until all input points~$P$ are covered, or otherwise until $\mathcal{R}$ has been exhausted. Recall that not all rectangles covering single points may be present due to the constraints, and even when they are, we may have chosen rectangles in $S$ that overlap with these singleton rectangles. Hence there are inputs for which we exhaust $\mathcal{R}$ before all points are covered.

\subsection{Comparison of MaxSAT and Greedy}
\begin{figure}
    \centering
    \includegraphics{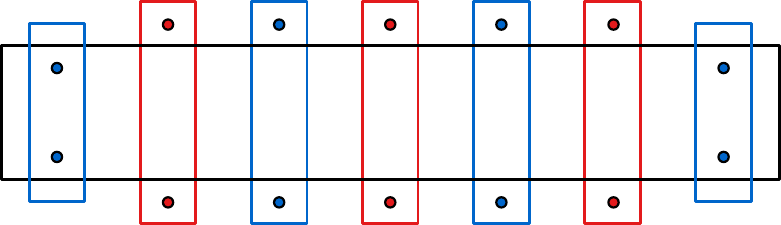}
    \caption{A point set with corresponding candidate set, consisting of $7$ vertically aligned candidate rectangles (blue/red), $1$ horizontally rectangle (black), and singletons (not drawn). }
    \label{fig: simple instance}
\end{figure}
In the following, we give a comparison of these two approaches with a simple example.
We presented an exact solver based on \textsc{MaxSAT} (see Section~\ref{subsec: exact}) and a greedy heuristic (see Section~\ref{subsec: greedy}) for the \RPFA~problem.
Both approaches take the same set of candidate rectangles as input and solve the \textsc{WIS} problem, either exactly or heuristically.
We illustrate these two procedures and their differences with an example instance; see Figure~\ref{fig: simple instance}.

In this example, the candidate rectangle set $\mathcal R$ consists of seven vertically stretched rectangles, one horizontal rectangle, and singletons for each point.
Recall that we assign to each candidate rectangle  $R \in \mathcal{R}$ the weight $2n|R|-1$, where $n= |P|$ and $|R|$ is the number of points covered by $R$.
With this weight assignment, the maximum-weight solution $S^*$ consists of the seven vertically aligned rectangles. 
Our exact solver method encodes these weights and pairwise overlaps of candidate rectangles as a weighted \textsc{MaxSAT} problem and then solves this problem exactly, thus finding the solution $S^*$.

Our greedy heuristic sorts the candidate rectangles in $\mathcal{R}$ by their weights in decreasing order. 
Then a set $S$ of pairwise disjoint rectangles is built greedily by choosing the rectangles in this order.
Precisely, this heuristic starts with the highest-weight rectangle $R_1$ and adds it to $S$. 
Then it checks the highest-weight rectangle $R_2$ in $\mathcal{R} \setminus \{R_1\}$. 
If this rectangle $R_2$ is entirely disjoint from all rectangles in $S$, we add it greedily to $S$; otherwise it is discarded.
This process is repeated until all points are covered by $S$ or the candidate set $\mathcal{R}$ is exhausted.

The horizontal rectangle $H$, which covers $4$ blue points, has the highest weight. 
Thus, the greedy approach collects $H$ as the first rectangle of the greedy solution $S$. 
All vertical candidate rectangles intersect $H$.
Thus, after picking $H$,  an uncovered point can only be covered by its singleton.
Hence, the greedy solution $S$ consists of $H$ and ten singletons for the points outside $H$.
Extending this simple construction by adding more vertical rectangles with appropriately colored points, the greedy approach could reach an approximation ratio arbitrarily close to  $2$.
However, we did not observe such behavior in our experiments.

\section{Experimental evaluation}\label{sec:evaluation}
We implemented the algorithms in Section~\ref{sec:algorithms} and first compared them against each other in terms of scalability as well as solution quality on synthetic data.
In the second part of the evaluation, we provide a case study in which we showcase how Worbels can be used on a real-world data set.

\subsection{Experimental Setup}
\label{sub:expsetup}
The experiments for measuring performance were run on a server equipped with two Intel Xeon E5-2640 v4 processors (2.4 GHz 10-core) and 160GB RAM, operating the 64-bit version of Ubuntu Bionic (18.04.2 LTS). The case study was carried out on a standard Windows laptop equipped with an Intel Core i7-6700HQ (2.60GHz, 4 Cores 8 Logical Processors) and 16GB RAM. The code was compiled with g++ 7.4.0 with Optimization level -O3. 

\mypar{Data sets}
We use three types of data for our experiments: two synthetic data sets that we use to analyze the scalability of our algorithms and the quality of their output, and a real-world data set to show the viability of our technique in practice.
The synthetic data sets, Uniform and Gaussian, are randomly generated point sets inside a bounding box of size 1000×1000 pixels, based on a uniform and a Gaussian distribution, respectively. The category labels in these synthetic instances are generated randomly such that each label has $3$ to $10$ letters. We vary the number of points~$n$ (from $20$ to $1000$, with increments of $10$), as well as the number of categories~$c$ ($2$, $3$, $4$, $8$, or $16$ categories).
\begin{description}
\item[Uniform] In the uniform model, we generate $n$ points inside a bounding box $\mathcal{B}$ of size 1000×1000 pixels uniformly at random.
    We generate $k$ random text strings of lowercase letters as category tags, each having a length between $3$ and $10$ characters, chosen uniformly at random.
    Then we assign one of the $k$ categories to each point, again uniformly at random.

\item[Gaussian] Intuitively, the points in a Gaussian instance are drawn from a mixture of $k$ Gaussian distributions.
	We first decide the number of points in each category $\{n_1, \dots, n_k\}$, using a Dirichlet distribution with dimension $k$. Note that the numbers in the categories sum up to $n$.
	For each category $i$, we use a separate Gaussian distribution whose mean~$\mu_i$ is sampled in the bounding box $\mathcal{B}$ uniformly at random.
	Then a value ~$\sigma_i$ is sampled uniformly between $0$ and $0.5$, and is used as the standard deviation for both dimensions.
	We then generate, for each $i=1, \dots, k$, $n_i$ points inside $\mathcal{B}$ using a Gaussian distribution with parameters $\mu_i$ and $\sigma_i$. 
\end{description}

Every five years, the U.S. Department of Agriculture (USDA) publishes a comprehensive summary of the census of agriculture, which includes harvest statistics for major crops in each county.
For the real-world data, US-Crops, we use statistical data about the crop harvest in the United States in the year 2007.\footnote{extracted from \href{https://www.arcgis.com/home/item.html?id=f6aa37a7376b4cbd8cc52abfdf8d63c4}{ArcGIS$\_$US$\_$CROPS$\_$2007}} The data is on county level, resulting in $3067$ data points, with a label from the set $\{\textsf{Corn}, \textsf{Wheat}, \textsf{Cotton}, \textsf{Soybeans}, \textsf{Vegetables}, \textsf{Various}\}$ indicating the dominant crop in the harvest of each county.
Additionally, we filter the counties with at least 100,000 acres harvested, to get a smaller data set with $817$ points, which we call US-Crops-filtered, and a variant of the data set where the aggregation is on state level, US-Crops-state, resulting in 48 data points. The point feature for each county/state uses the latitude and longitude of a central location in each county/state as coordinates. 
\begin{description}
\item[Real-world] 
We generate three data sets from this data base, where we assign the labels representing the dominant crop to each county/state.
\begin{description}
    \item[US-Crops]
    The full data extracted from the data base. 
    Each point feature represents a county in the data. 
    This data set consists of $3067$ counties, partitioned into six categories, $\{\textsf{Corn}, \textsf{Wheat}, \textsf{Cotton}, \textsf{Soybeans}, \textsf{Vegetables}, \textsf{Various}\}$. The label \textsf{Various} is assigned, if there is no data for that county in the census; any crop may be dominant.

    \item[US-Crops-filtered]
    In this middle-size data set, we filter all counties with total area of harvested crops larger than 100,000 acres, resulting in a total of $817$ features with five categories, all previous categories, minus the label \textsf{Various}, since those states will always be filtered out.

    \item[US-Crops-state] In this state-level data set, each point feature represents one state. 
    We compute the dominant crop per state by summing up acres of the counties in a state, and assigning the crop with the highest total harvest as the dominant crop in the state.
    Then, we use the latitude and longitude obtained from \href{https://www.latlong.net/category/states-236-14.html}{Geographic Coordinate of US States} as coordinates for a point feature representing the state.
    This instance has $48$ points and each point is assigned one of five labels, similar to the filtered data set.
\end{description}
\end{description}


\subsection{Experimental Results}
\label{subsec: Eresults}
Here we explain the results from a selection of experiments, 
using the synthetic data sets.

\mypar{Instance parameters vs performance}
We first show that the running time of both the exact solver and the greedy heuristic are greatly affected by the size of the candidate set~$\mathcal{R}$, more so than by the number of points~$n$ or number of categories~$c$. 
In this experiment, both algorithms were executed with parameter settings $\rho_l=0$, $\rho_u=\infty$, and $t=\rho_t=f=0$. 
\begin{figure}[tbp]
	\captionsetup[subfigure]{justification=centering}
	\begin{subfigure}[t]{0.48\columnwidth}
	    \centering
		\includegraphics[scale=0.5]{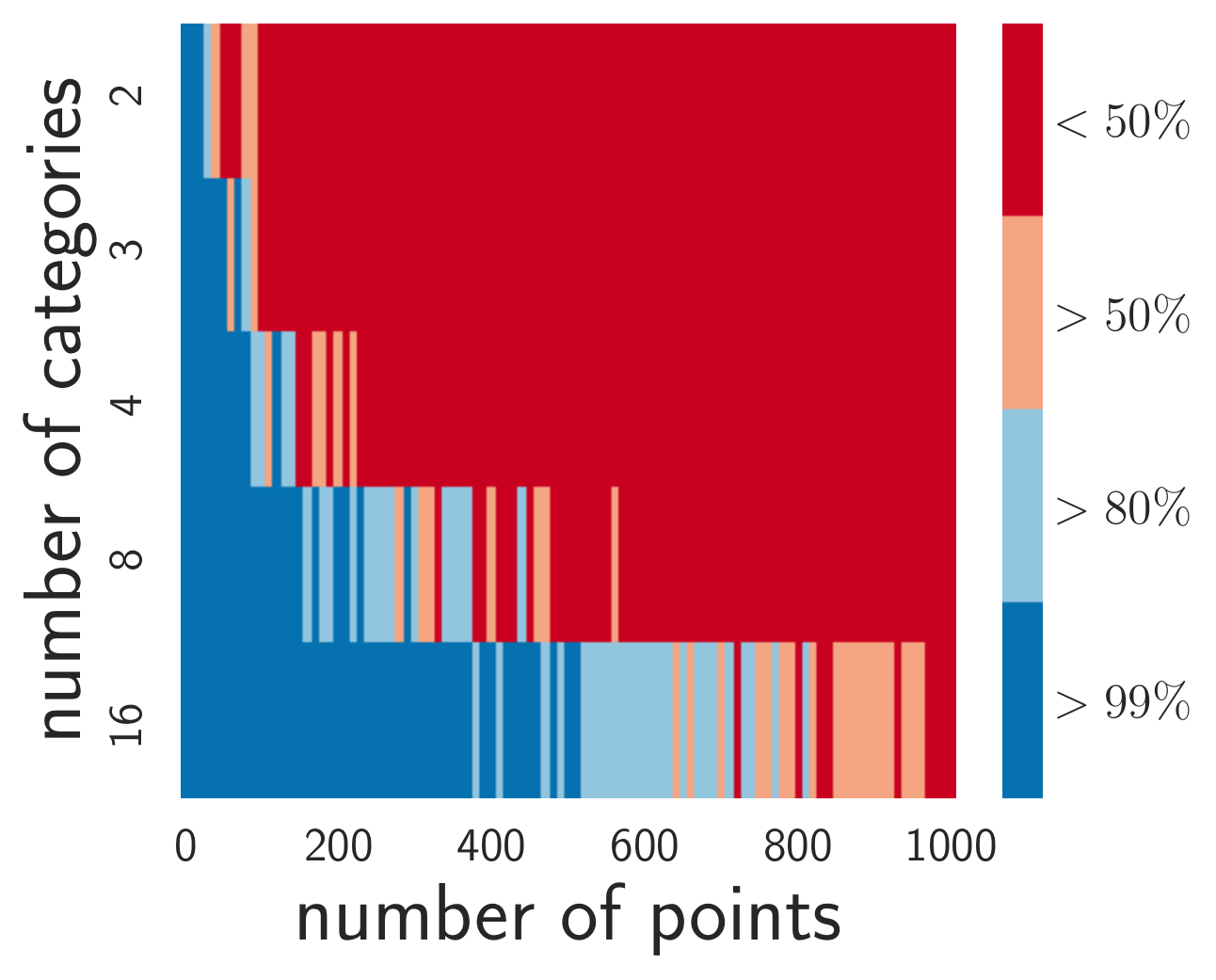}
		\caption{}
	\end{subfigure}
	\hfill
	\begin{subfigure}[t]{0.48\columnwidth}
	    \centering
	    \includegraphics[scale=0.5]{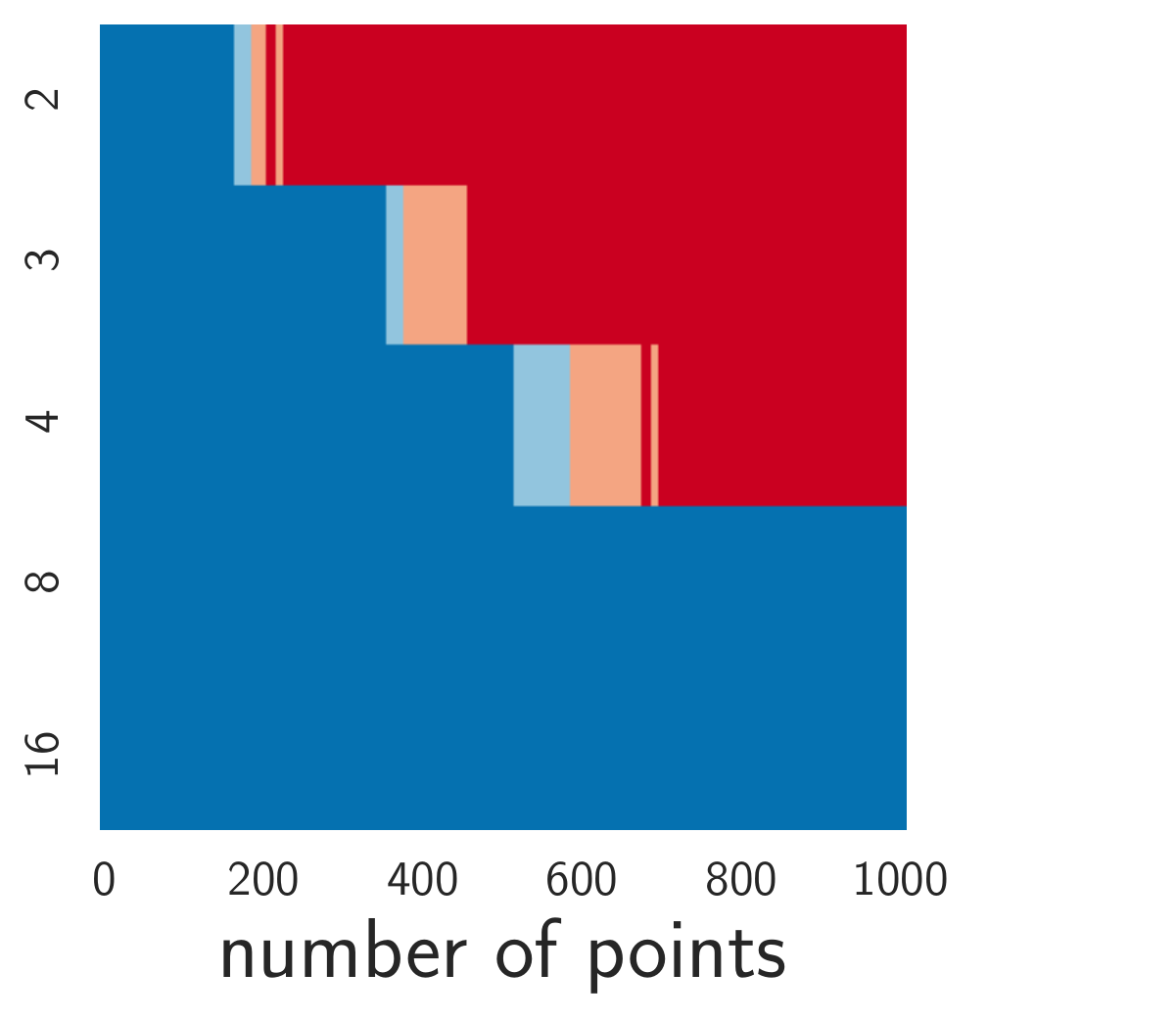}
	    \caption{}
    \end{subfigure}

    \caption{Percentage of instances, for which the exact solver terminates within 30 minutes. \textsf{\textbf{(a)}}~Gaussian, and \textsf{\textbf{(b)}}~Uniform instances with $n= 20$ to $1000$ points, $c \in \{2,3,4,8,16\}$ categories. 
    }
	\label{fig:exact_vs_greedy}
\end{figure}

Figure~\ref{fig:exact_vs_greedy} shows the percentage of instances of Gaussian and Uniform data sets for which the \textsc{MaxSAT} solver terminates within 30 minutes, for all combinations of $n$ and $c$. As $n$ grows, we see that the percentage of failed runs grows as well, but not very strongly when $c$ is large ($c=16$ for Gaussian, $c=8$ and up for Uniform). However, when $c$ is smaller, especially $c=2$, less than half of the instances finish within half an hour, even for smaller point sets, such as $n=300$. To understand this discrepancy, we turn to Figure~\ref{fig:exact_vs_greedy_time}.


Figure~\ref{fig:exact_vs_greedy2} shows the relation between the size of the candidate set, and the running time, for each individual Gaussian instance. The $y$-axis ends at a running time threshold of 30 minutes, and we plot a red point at 30 minutes when a run did not finish. This happens only for the exact solver, with the exception of a single run for the greedy approach that ran out of memory. As the candidate set grows, we see the running time of the exact solver growing much more quickly than the running time of the greedy heuristic: while the greedy approach can still find a solution in under a second for most instances where the candidate set consists of around 100,000 rectangles, the exact solver cannot find a solution within 30 minutes in almost all cases even for 6500 rectangles. We stopped running the exact solver when candidate sets exceed 100,000 rectangles. Note that the 30 minute timeout was chosen for practical reasons only, and many instances which produce more than 5000 rectangles in the candidate set may take much longer than 30 minutes to finish using the exact solver. This prevents the exact solver from being a viable option for all but the smallest instances in practice. On the other hand, the greedy heuristic can still find solutions within 30 minutes, even for candidate sets of up to 100 million rectangles. 
\begin{figure}[tbp]
\captionsetup[subfigure]{justification=centering}
\begin{subfigure}[t]{0.47\columnwidth}    
	    \centering
	    \includegraphics[width=\columnwidth]{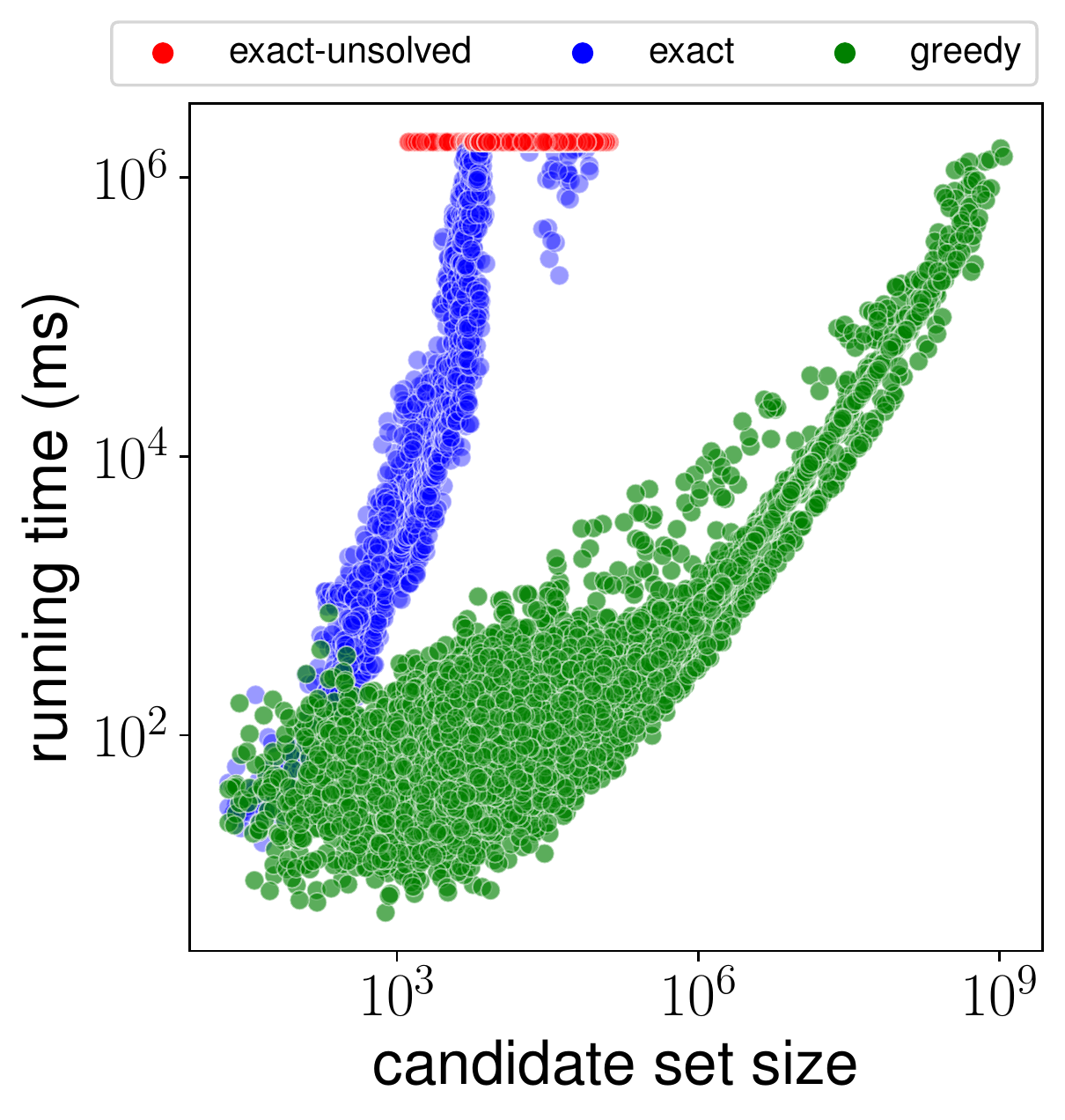}
	    \caption{}
	    \label{fig:exact_vs_greedy2}
    \end{subfigure}
    \begin{subfigure}[t]{0.47\columnwidth}  
        \centering
    	\includegraphics[width=\columnwidth]{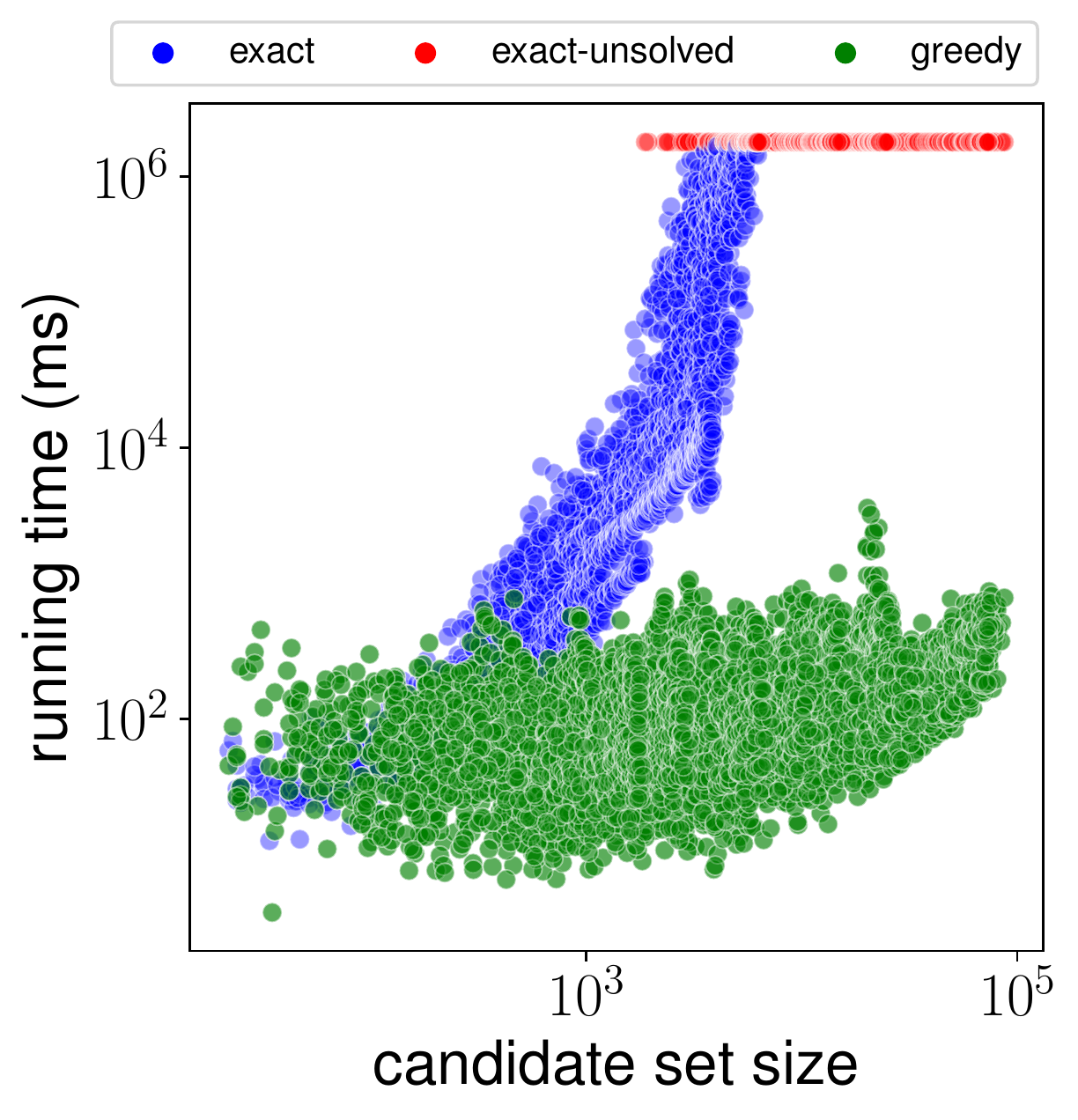}
	    \caption{}
	    \label{fig:exact_vs_greedy_uni_time}
    \end{subfigure}
	\caption{Log-log plot of running time with respect to candidate set size of \textsf{\textbf{(a)}}~Gaussian and \textsf{\textbf{(b)}}~Uniform instances. Red points show instances that failed to finish in 30 minutes, for candidate sets up to $10^5$. On a vertical line, red and blue points sum up to an equal number of green points.}
	\label{fig:exact_vs_greedy_time}

\end{figure}

The results for the Uniform data sets are very similar; see  Figure~\ref{fig:exact_vs_greedy_uni_time}.
The running time of the greedy heuristic grows very slowly as the candidate set size increases. On the other hand, the exact solver again shows a steep increase in running time, as we also saw for Gaussian instances (Figure~\ref{fig:exact_vs_greedy2}).
Similar as for the Gaussian instances, the greedy heuristic shows a significantly faster running time, already for small instance size of $1000$ candidate rectangles, while the heuristic also solves large instances even with up to $100\,000$ rectangles. 

We believe that the growth of the candidate set, for high values of $n$ and especially for small values of $c$, can be explained as follows. During the computation of the candidate set, fewer categories lead to many strips with base rectangles that can be extended often, because the probability of encountering a point of a different color is lower. These cases should occur even more often in Gaussian than in Uniform, and we indeed find that the number of instances that finishes in $30$ minutes is lower for Gaussian, judging by Figure~\ref{fig:exact_vs_greedy}. Additionally, we can confirm that the candidate sets for Uniform instances never exceed $100,000$ rectangles.


\begin{figure}[tbp]
\captionsetup[subfigure]{justification=centering}
    \centering
    \begin{subfigure}[t]{0.47\columnwidth}
    \includegraphics[width=\columnwidth]{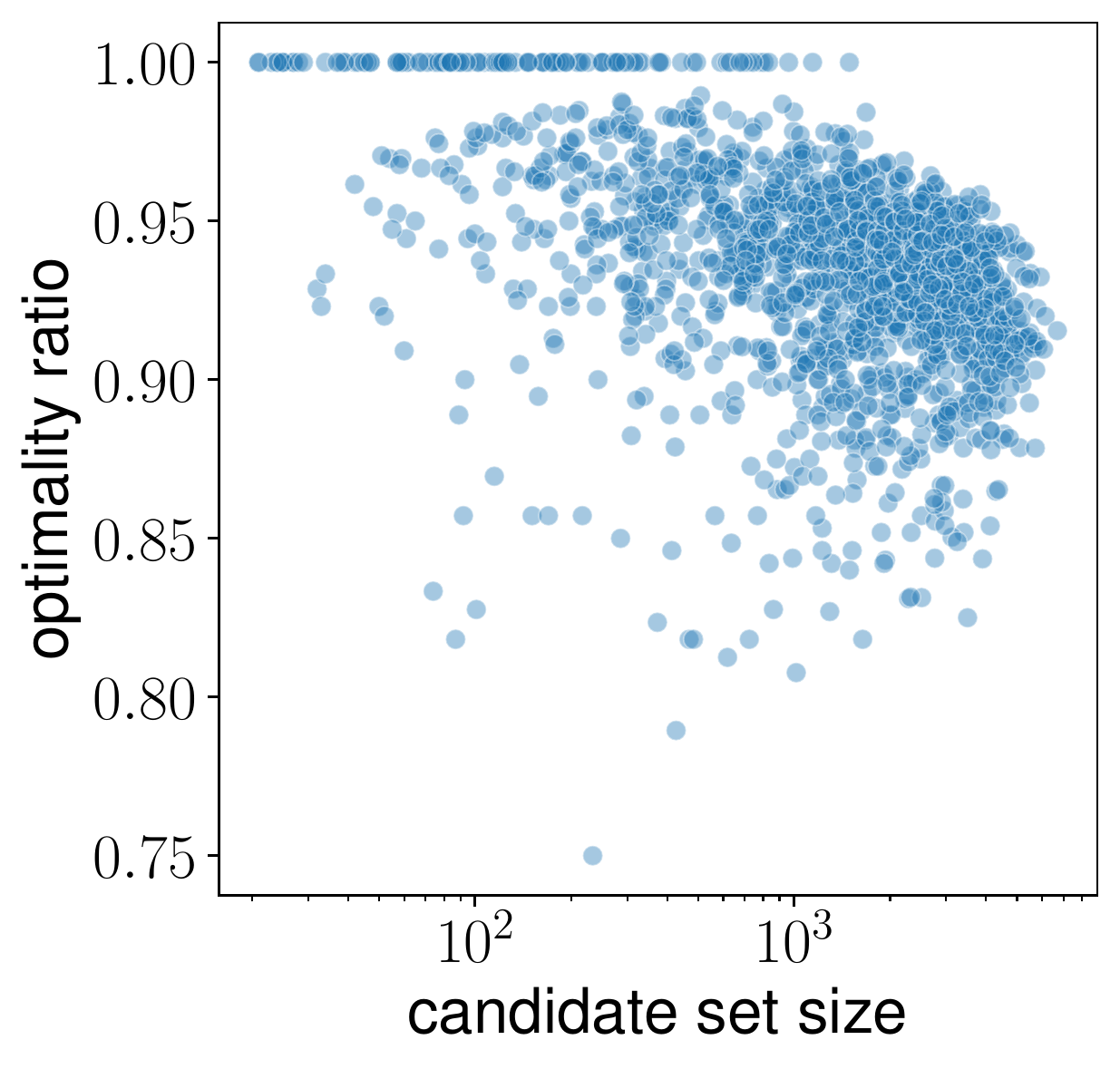}
	\caption{}
	\label{fig:exact_vs_greedy3}    
    \end{subfigure}
    \hfill
    \begin{subfigure}[t]{0.47\columnwidth}
    \centering
    \includegraphics[width=\columnwidth]{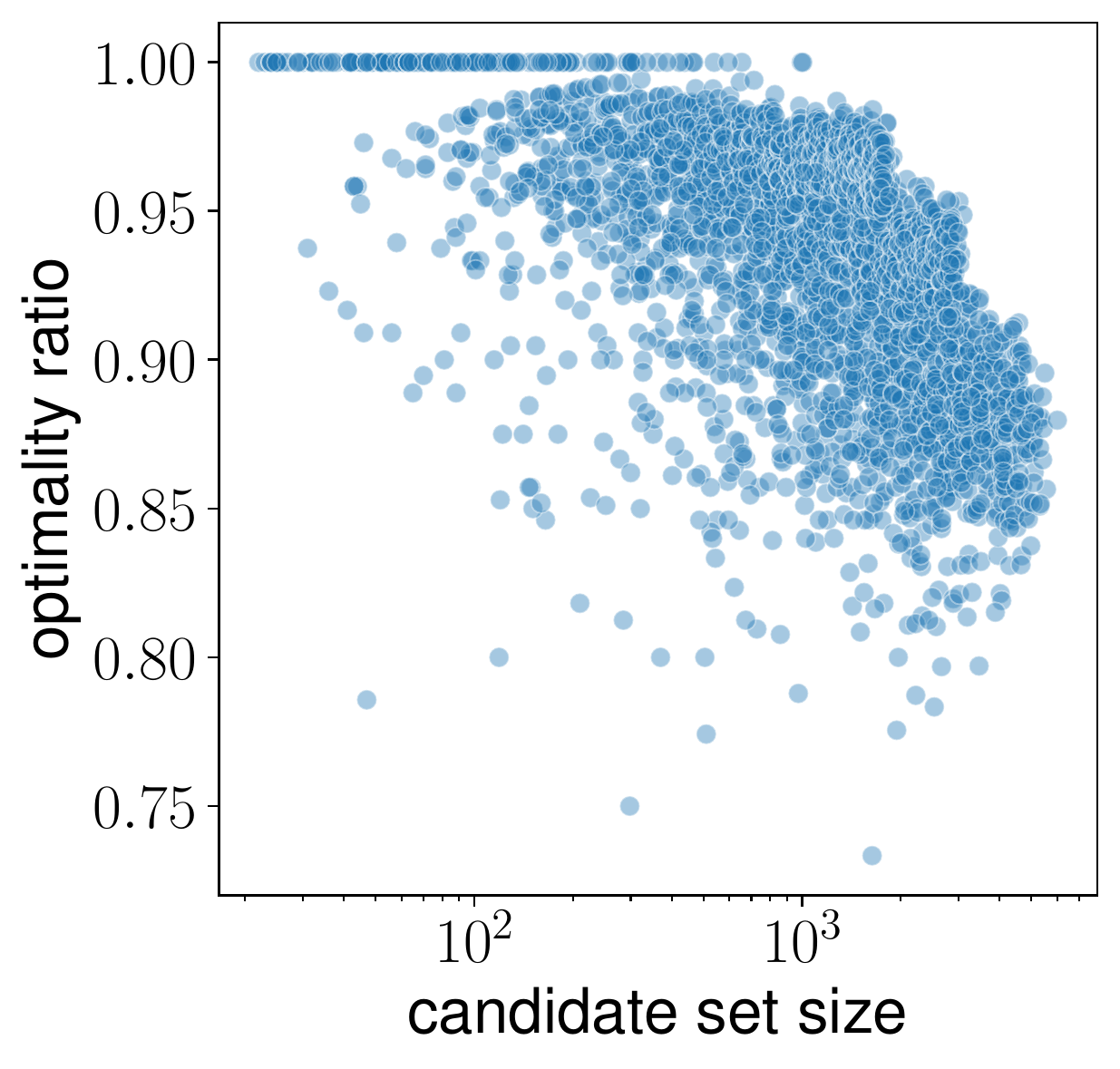}
	\caption{}
	\label{fig:exact_vs_greedy_uni_ratio}
    \end{subfigure}
	\caption{Ratio of the solution size between greedy and optimal solutions for \textsf{\textbf{(a)}}~Gaussian and \textsf{\textbf{(b)}}~Uniform instances with various candidate set sizes.}
\end{figure}

\mypar{Candidate set size and optimality}
Now that we have established that the exact solver, unlike the greedy heuristic, is not viable for larger instances, we proceed to verifying whether the heuristic finds near-optimal solutions. We therefore measure the ratio between the number of rectangles in a solution produced by the greedy heuristic and an optimal solution for a given candidate set. The optimal solution size is found by the exact solver, and hence we can only do this analysis for data sets where the exact solver actually terminates.

In Figure~\ref{fig:exact_vs_greedy3} we show this ratio for all Gaussian instances where the exact solver found a solution. We see that the solutions found by the greedy approach deviate more from optimality, as the candidate set size increases. However, for the majority of evaluated instances, the solution size computed by the greedy heuristic is at most 10\% larger than the optimal solution, and there are only two instances where this difference grows to 20\%.

Just as for the Gaussian instances, we consider only the small Uniform instances, for which the exact solver found a solution.
The plot (Figure~\ref{fig:exact_vs_greedy_uni_ratio}) shows that the greedy heuristic consistently finds solution sizes higher than $70\%$ of the optimum solution size. 
For most of the evaluated instances, the size of the greedy solution is at least $85\%$ of the optimum solution size.
For small instances with at most $100$ candidate rectangles, the greedy solutions reach the optimum sizes for most of the instances.
As expected, the optimality ratios are distributed more widely as the candidate set size increases. 
This again mirrors the same findings for the Gaussian instances .

We also see that the greedy approach is able to find optimal solutions for some instances. This means that our greedy heuristic is not only much faster than the exact solver, but also produces solutions that are still of high quality in terms of solution size. We therefore focus on the solutions produced by the greedy heuristic in our case study.
Overall, the performance of the greedy heuristic shows a well-balanced compromise between running time and solution quality on these instances.

\subsection{Case Study -- U.S. Crop Harvest}
We use our greedy heuristic to produce Worbel visualizations for the US-Crops data set, to see the effectiveness of the visualizations on real-world data. These pictures can easily be produced on a less powerful machine (see Section~\ref{sub:expsetup}).
We used the following parameters to compute the visualizations in this case study (unless mentioned otherwise), $\rho_l=0.75$, $\rho_u=2$, $t=2$, $\rho_t=0.2$, and $f=16$.

We start by considering the full data set, US-Crops, comprising 3076 points. The result is shown in Figure~\ref{fig:teaser} and was computed in 98 seconds. If we look at the Worbel on the map of the U.S. (Figure~\ref{fig:teaser:a}), we see a fairly dense set of labels that leaves few open spaces. Around the sides of the map, we see that labels can extend outwards. The algorithm clusters a few points or assigns labels to singletons/outliers in these cases, resulting a labeling that resembles a classic point feature labeling. Along the east coast, we see many such labels. However, in the middle of the map, many points are clustered and covered by single label. These labels resemble an area labeling of the map, indicating that the areas covered by labels predominantly have points of that same category. For example, this happens with the labels \textsf{Wheat} and \textsf{Corn}, around Montana/North Dakota, and Iowa/Wisconsin/Illinois, respectively. Overall, solving the \RPFA\ problem here seems to result in a labeling that is a nice hybrid between area labeling and (single) point labeling.

If we look at the computed rectangles and the data points that are not covered (Figure~\ref{fig:teaser:b}), we see that most rectangles stack nicely and cover the represented points as intended. The minimum font size prevents some points from being covered; either because their label would intersect some other label in the Worbel (see, e.g., the isolated points in the west of the U.S.), or because their label would intersect too many points with a different label (such as in Alabama/Georgia area and in Oklahoma).


Our findings for US-Crops-filtered are similar. 
Figure~\ref{fig:crops-middle} shows the Worbel for this data set, computed in about 15 seconds. 
Since there are less points, and fewer outliers in this data set, we see larger labels throughout this Worbel. 
The open areas did not contain any points, because of the filtering, and the more densely packed areas leave some points uncovered, again due to the minimum font size. 
We see a similar trend as before, where areas with many points result in tightly packed area labelings, while for sparser areas a single point labeling is produced. 
Also note that, while the parameters for aspect ratio are not very strict, we still produce rectangles close to the aspect ratio of the respective labels.

\begin{figure}
    \captionsetup[subfigure]{justification=centering}
	\begin{subfigure}[t]{0.47\textwidth}
		\includegraphics[width=\textwidth]{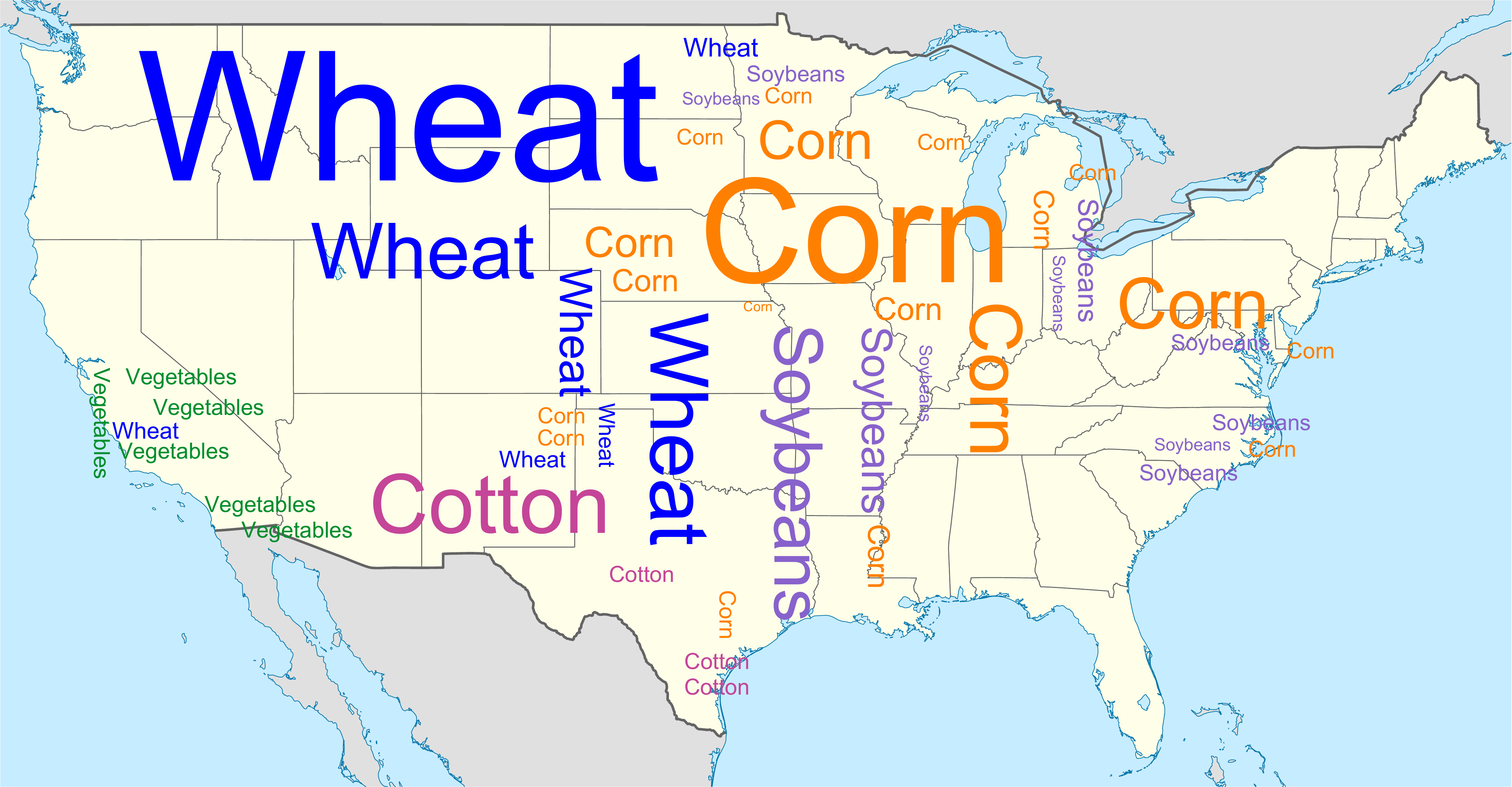}
		\caption{}
	\end{subfigure}
	\begin{subfigure}[t]{0.47\textwidth}
	    \includegraphics[width=\textwidth]{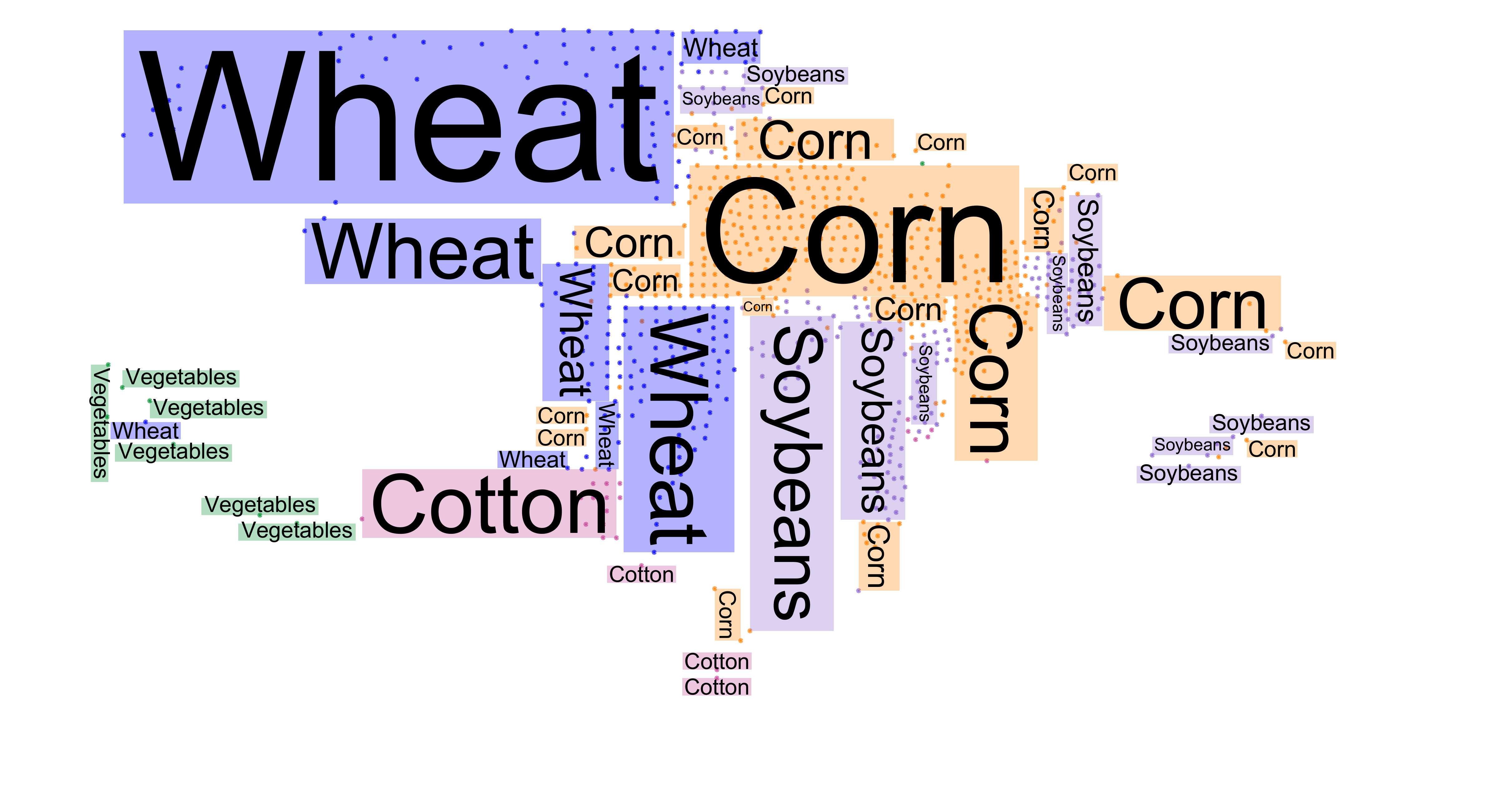}
	    \caption{}
    \end{subfigure}
	\caption{Worbels of the US-Crops-filtered data set, using parameters $\rho_l=0.75$, $\rho_u=2$, $t=2$, $\rho_t=0.2$, and $f=16$, and produced in 15 seconds. \textsf{\textbf{(a)}}~Drawn on a map. \textsf{\textbf{(b)}}~Drawn over the data points; rectangles outline the computed area per label.}
	\label{fig:crops-middle}
\end{figure}

In Figure~\ref{fig:crops-states} we show a Worbel for the US-Crops-State data set, produced in less than a second. This data set contains a single data point per state, and hence is very sparse. In Figure~\ref{fig:crops-states-1} we see the Worbel generated with the input parameters as before. However, since there are very few points, we decided to tune down the tolerance for misrepresentation to zero, since the error we make is significantly higher than for the larger data sets. Setting $t=0$ produces the Worbel in Figure~\ref{fig:crops-states-2}. 
Even though there are only 48 points, we get a clustering of \textsf{Wheat}, \textsf{Corn}, and \textsf{Soybeans} in the central states. Along the boundary of the map, states are labeled as singletons, so the produced Worbel is still a nice hybrid between area labeling and point labeling, even for sparse data.

\begin{figure}[b]
	\begin{subfigure}[t]{0.47\textwidth}
		\includegraphics[width=\textwidth]{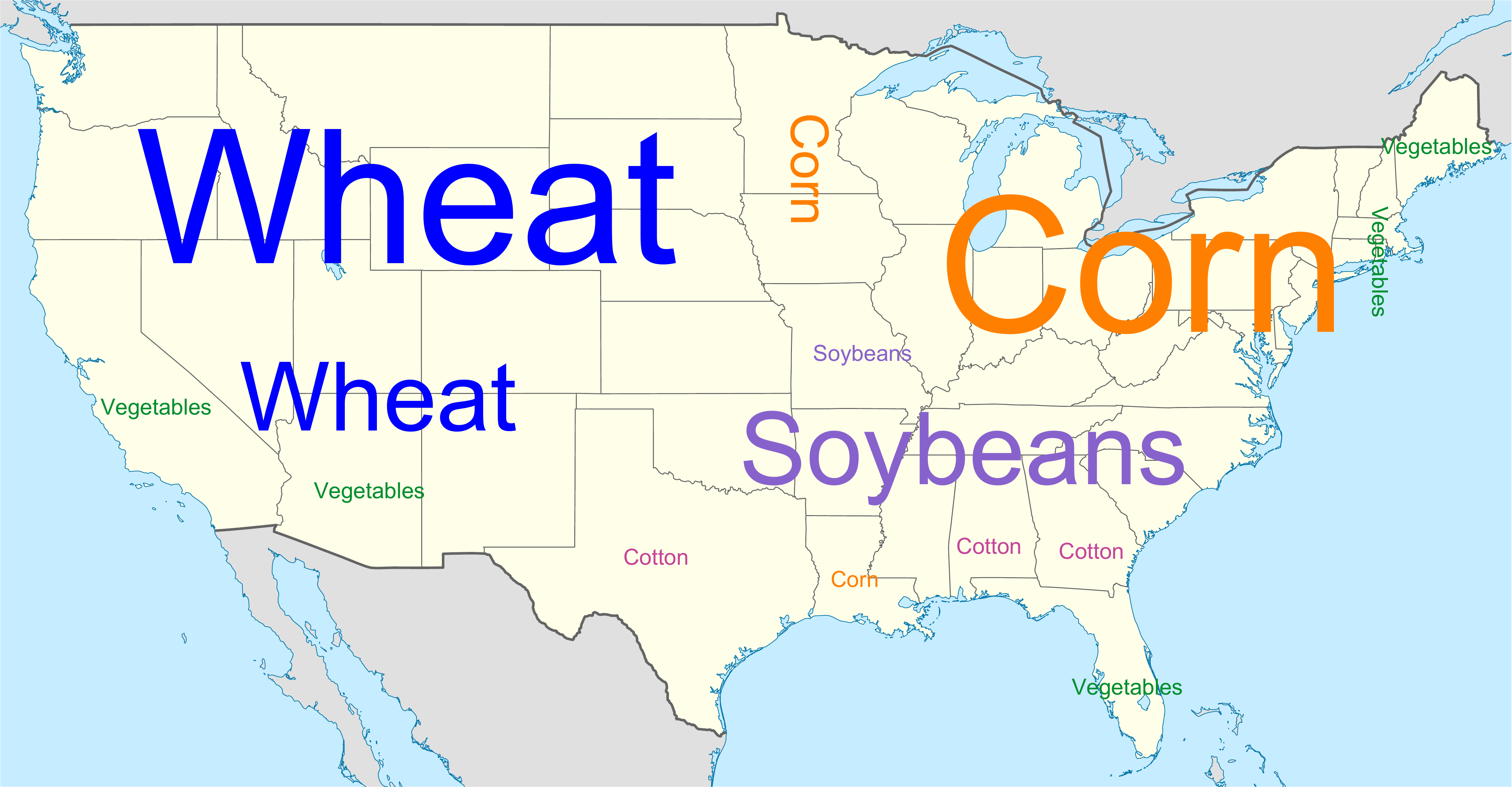}
		\caption{Worbel with tolerance for misrepresentation}
			\label{fig:crops-states-1}
	\end{subfigure}
	\hfill
	\begin{subfigure}[t]{0.47\textwidth}
	    \includegraphics[width=\textwidth]{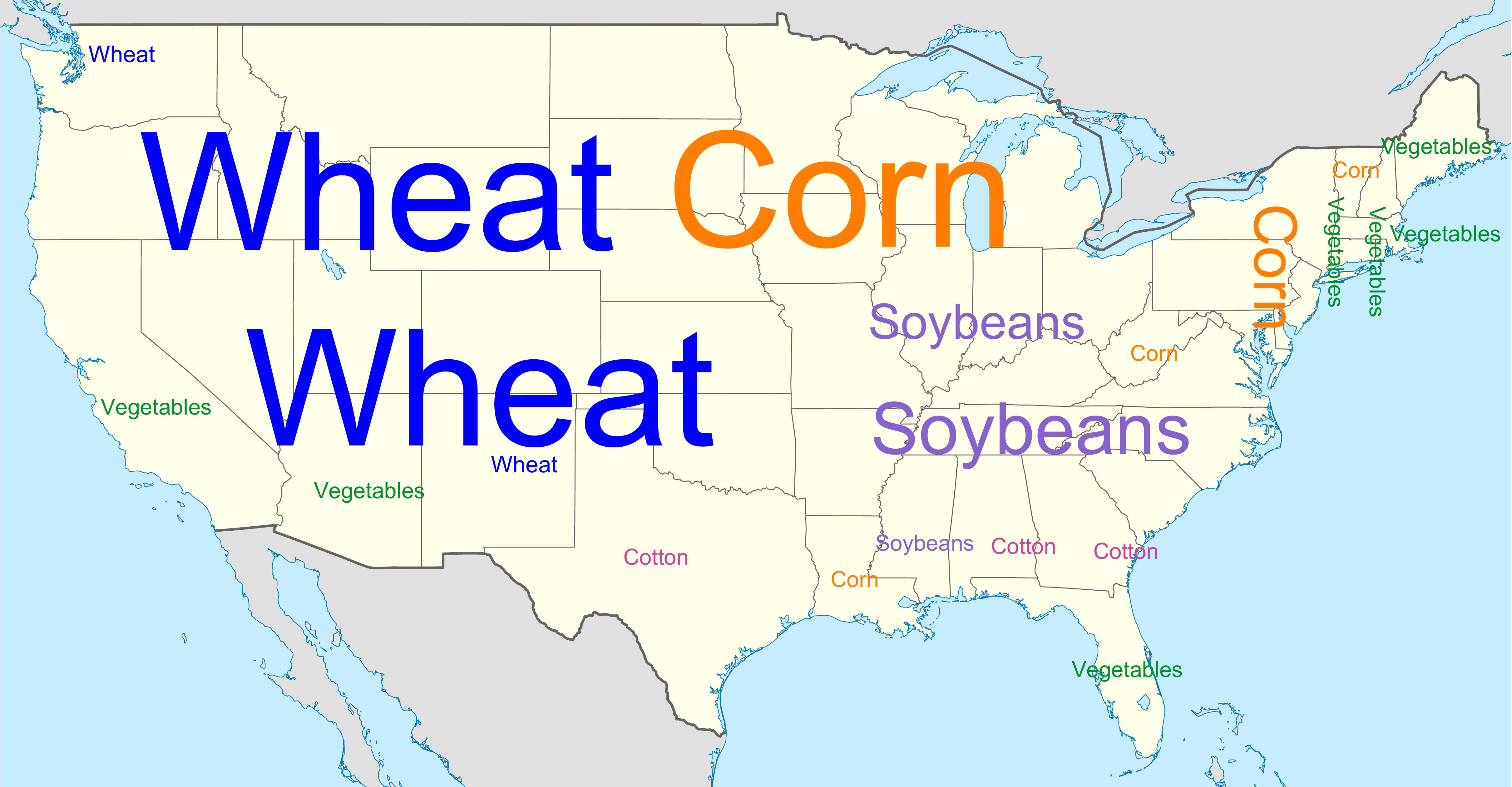}
	    		\caption{Worbel without misrepresentation}
	    	\label{fig:crops-states-2}
    \end{subfigure}

	\caption{Worbels of the US-Crops-state data set, overlaid on a map, and produced in $<$1 second. \textsf{\textbf{(a)}}~Using parameters $\rho_l=0.75$, $\rho_u=2$, $t=2$, $\rho_t=0.2$, and $f=16$. \textsf{\textbf{(b)}}~Using parameters $\rho_l=0.75$, $\rho_u=2$, $t=0$, $\rho_t=0.2$, and $f=16$}
	\label{fig:crops-states}
\end{figure}

\section{Discussion and Concluding Remarks}\label{sec:conclusion}
The Worbel hybrid visualization scheme proposed in this article is designed to provide an aggregated textual labeling for geospatial data points. Our contributions include a formalization of the problem underlying the aggregation of geospatial data via the \RPFA\ problem, and a proof establishing that \RPFA\ is \NP-complete. 

In view of this algorithmic lower bound, we propose and implement two algorithms for computing Worbels. The first is designed to find theoretically optimal Worbels for small to medium-size instances of \RPFA, and is based on a \textsc{MaxSAT} model (which turned out to be significantly more scalable than ILP and CSP-based approaches in our initial experiments). 
We complemented this with an iterative heuristic algorithm. Our experiments showed that the running time performance of this heuristic outscales the exact \textsc{MaxSAT} approach by several orders of magnitude and can produce Worbels for instances with thousands of data points in a matter of seconds, even on standard work stations. At the same time, for the vast majority of tested instances the Worbels produced by the heuristic only exhibit an overhead of at most 10\% additional aggregate labels.

The experimental results also give rise to the practically interesting option of combining the strengths of both approaches. It turned out that the size of the candidate set of rectangles is a much better predictor of the computation time than the size of the point set or the number of different point categories alone. Since both methods begin by computing these candidate sets, we can compute them and then depending on the resulting size either invoke the \textsc{MaxSAT} solver (if we expect fast termination) or invoke the greedy heuristic (if the candidate set size exceeds an adjustable threshold). Further research should investigate more sophisticated methods to prune the candidate set by discarding rectangles which are identified as suboptimal without compromising the overall solution quality, and thus accelerate computation times for both algorithms.


An interesting observation from our case study on the US crop data is that the resulting Worbels actually combine features of area and point labeling. In homogeneous groups of points with the same label, a single large rectangle can be placed, effectively labeling the entire covered area. In areas with more heterogeneous point features, labels for smaller clusters or even singleton points are placed, much as in traditional point feature label placement. 
As a result, Worbels implicitly decrease the information density of text-based data visualizations on maps in areas of low entropy by placing fewer labels, while showing more detailed and less aggregated information in areas of higher entropy by placing more but smaller labels.
An important open question is whether these promising properties of Worbels translate into quantifiable usability and readability advantages when compared to maps with individual point labels and/or less precise geo word clouds or tag maps.

Ultimately, we believe that integrating Worbels into interactive visual analytics tools or GIS systems provides an added value for the users of these systems. One opportunity in this context is that  analysts can modify and improve an initially computed Worbel by locally fine-tuning quality parameters or adding context-dependent constraints and thus creating Worbels according to their individual preferences in a human-in-the-loop process. This will require more advanced algorithms for computing visually stable Worbels on a dynamic input.

\bibliographystyle{plainurl}
\bibliography{lit_gpl_arxiv.bib}

\end{document}